\documentclass[10pt, conference, compsocconf]{IEEEtran}
\usepackage[utf8]{inputenc}

\usepackage[english]{babel}
\usepackage[utf8]{inputenc}
\usepackage{amsmath}
\usepackage{graphicx}
\usepackage{boxedminipage}
\usepackage{epsfig}
\usepackage{psfrag}
\usepackage{url}
\usepackage{comment}
\usepackage{afterpage}
\usepackage{cite}
\usepackage{pstricks}
\usepackage{graphics}

\usepackage{mdwtab}
\usepackage{wrapfig}
\usepackage{graphics}

\usepackage{algorithm}  
\usepackage{algorithmic} 

\usepackage{amssymb} 

\usepackage{makecell}

\usepackage{tabularx}

\usepackage{epsfig}
\usepackage{fixfoot}
\usepackage{listings}
\usepackage{psfrag,wrapfig}
\usepackage{xspace,soul}
\usepackage[colorinlistoftodos]{todonotes}

\usepackage{etex}

\usepackage[font=small]{caption}
\usepackage[font=footnotesize]{subcaption}

\newcommand{\ATTCK}{{\sffamily ScheduLeak}\xspace}

\newcommand{\ATTCKNF}{{ScheduLeak}\xspace}

\newcommand{\OBSERVER}{{observer task}\xspace}

\newcommand{\NK}{{\vert J^{\tau_i}_k \vert}\xspace}
\newcommand{\NKSET}{{J^{\tau_i}_k}\xspace}

\newcommand{\CK}{{l_k}\xspace}

\newcommand{\ATME}{{\sigma^{\tau_i}_{k,h}}\xspace}
\newcommand{\ATMEW}{{\overline{\sigma^{\tau_i}_{k,h}}}\xspace}

\newcommand{\STME}{{\delta^{\tau_i}_{k,h}}\xspace}

\newcommand{\eg}{{\it e.g.,}\xspace}

\newcommand{\ie}{{\it i.e.,}\xspace}

\newcommand{\ci}{{\it (i) }}
\newcommand{\cii}{{\it (ii) }}
\newcommand{\ciii}{{\it (iii) }}
\newcommand{\civ}{{\it (iv) }}
\newcommand{\ca}{{\it (a) }}
\newcommand{\cb}{{\it (b) }}
\newcommand{\cc}{{\it (c) }}

\usepackage{amsthm}
\theoremstyle{definition}
\newtheorem{definition}{Definition}
\newtheorem{example}{Example}

\newtheorem{lemma}{Lemma}
\newtheorem{theorem}{Theorem}

\usepackage{cleveref}
\crefformat{section}{\S#2#1#3}

\usepackage{multirow}

\usepackage{mathtools}
\DeclarePairedDelimiter{\ceil}{\lceil}{\rceil}

\renewcommand{\baselinestretch}{0.98} 

\begin{document}


\title{\Large A Reconnaissance Attack Mechanism for Fixed-Priority Real-Time Systems
\vspace{-2ex}}


\author{
\IEEEauthorblockN{\small Chien-Ying Chen\IEEEauthorrefmark{1}, AmirEmad Ghassami\IEEEauthorrefmark{2}, Sibin Mohan\IEEEauthorrefmark{1}, Negar Kiyavash\IEEEauthorrefmark{2},  Rakesh B. Bobba\IEEEauthorrefmark{3}, Rodolfo Pellizzoni\IEEEauthorrefmark{4} and Man-Ki Yoon\IEEEauthorrefmark{1}} 
\IEEEauthorblockA{\small \IEEEauthorrefmark{1}Dept. of Computer Science, \IEEEauthorrefmark{2}Dept. of Electrical and Computer Engineering, University of Illinois at Urbana-Champaign, Urbana, IL, USA}
\IEEEauthorblockA{\small \IEEEauthorrefmark{3} School of Electrical Engineering and Computer Science, Oregon State University, Corvallis, OR, USA}
\IEEEauthorblockA{\small \IEEEauthorrefmark{4}Dept. of Electrical and Computer Engineering, University of Waterloo, Ontario, Canada}
\small Email: \{\IEEEauthorrefmark{1}cchen140,
\IEEEauthorrefmark{2}ghassam2,
\IEEEauthorrefmark{1}sibin,
\IEEEauthorrefmark{2}kiyavash,
\IEEEauthorrefmark{1}mkyoon\}@illinois.edu,
\IEEEauthorrefmark{3}rakesh.bobba@oregonstate.edu,
\IEEEauthorrefmark{4}rodolfo.pellizzoni@uwaterloo.ca
\vspace{-2ex}}

\maketitle

\begin{abstract}
In real-time embedded systems (RTS), failures due to security breaches can cause serious damage to the system, the environment and/or injury to humans. Therefore, it is very important to understand the potential threats and attacks against these systems.
In this paper we present a novel reconnaissance attack
that extracts the exact schedule of real-time systems designed 
using
fixed priority scheduling algorithms. 
The attack is demonstrated on both a real hardware platform and a simulator, with a high success rate. 
Our evaluation results show that the algorithm is robust even in the presence of execution time variation.

\end{abstract}

\section{Introduction}
\label{sec::intro}
Today, real-time embedded systems are used in a variety of domains such as automotive, medical devices, avionics and spacecraft and industrial control (\eg power plants, chemical plants) to name but a few. Many of these systems are {\em safety-critical} by nature. That is, any problems that deter from the normal operation of such systems could result in damage to the system and the environment, and even threaten human safety. 

Conventionally, real-time systems (RTS) used custom platforms, software and protocols and were normally not linked to the external world. As a result, security was typically not a priority in the design of RTS. However, due to the drive towards remote monitoring and control facilitated by the growth of the Internet and increased connectivity, as well as the rise in the use of common-off-the-shelf (COTS) components, these traditional assumptions are increasingly being challenged. This is evident from recent work (\eg \cite{checkoway2011comprehensive,gollakota2011they}) that showed that it is easy to attack such systems. The work also demonstrated that successful attacks against such systems could lead to problems more serious than just loss of data or availability because of their critical nature.
This has spurred work on security of RTS, including both defensive techniques (\eg \cite{embeddedsecurity:mohans3a2013,yoonsecurecore2013,Zadeh:2014,embeddedsecurity:mohan2014,embeddedsecurity:mohan2015,YoonMHM:2015,xie2007schedulesecurity,lin2009rtssecurity}) and attacks (\eg \cite{checkoway2011comprehensive,gollakota2011they}). These attacks for the most part leveraged the lack of authentication in RTS and their communications. 

In this paper, we present a \emph{novel reconnaissance attack on RTS} that exploits their predictable nature to gather system behavior information that can be used to launch other attacks. Reconnaissance is often the first and an important step for many attacks.
For example, the attackers in the case of Stuxnet~\cite{chen2011lessons} were in the system undetected for months to collect information before launching their attacks. In RTS, any deviation from the expected behavior is suspicious and easier to detect in contrast to general purpose systems; attackers must operate within narrow operational parameters (\eg stringent timing and resource constraints) if they are to avoid detection. 
This is particularly true in the case of side-channel and covert-channel attacks. Often such attacks require {\em precise knowledge} of when the victim task is about to execute for the attacker to achieve high probability of success \cite{Kocher:1996,Aciicmez:2007:PSB,Liu:2015:LLC}. 
In such a case, the attacker may have the code of the system, but does not know when the tasks are activated (\ie task offsets are unknown) because the attacker may enter the system sometime after startup.
Hence, the ability to {\em accurately reconstruct} the behavior of the system -- in this context extracting the points in time when one or more tasks of interest (victims) execute, without being detected, would significantly improve the success of such attacks on RTS. 
In this work we present algorithms to {\em extract\footnote{We use {\em extract} and {\em reconstruct} interchangeably in this paper.} (and hence, leak) the precise schedule of the system} (\cref{sec::approach}) without perturbing the system behavior (to reduce the risk of detection)\footnote{A shorter preliminary version of this paper~\cite{certs2016shceduleak} was presented at a workshop, without archived proceedings.}.

Due to their safety-critical nature, RTS are designed with great care and significant engineering effort to operate in a {\em predictable} manner. For instance, 
\ca designers take great care to ensure that the constituent tasks in such systems execute in an expected manner~\cite{LiuLayland1973}, 
\cb their interrupts are carefully managed\cite{Zhang:2006:PIS} 
and \cc the execution time 
is analyzed to great degree at compile time and run-time (\eg 
\cite{hansen2009statistical, 
burns2000predicting,
cazorla2013proartis}).
However, as we demonstrate, this predictability can be a double edged sword --
it
often makes it easier for adversaries to gauge the
behavior of the system with high precision. 
This can increase the success
rates for certain attacks, \eg side-channel\cite{Kocher:1996} and covert-channel
attacks\cite{embeddedsecurity:volp2008}. 
An example that demonstrates a cache-based side-channel attack,
leveraging this predictable nature,
is presented in \cref{subsec:example_case_study}.

Our focus here is on hard real-time systems designed around {\em fixed
priority algorithms}\cite{Buttazzo:1997:HRC,LiuLayland1973} since 
these are the most common class of real-time scheduling algorithms found in practice today\cite{Liu:2000:RTS}. Also, this class of scheduling algorithms is most
vulnerable to attacks since they 
are very predictable. We carry out exhaustive simulations to understand the design space for the proposed algorithms (\cref{subsec:simulation_eval}). 
Furthermore, we demonstrate an {\em actual implementation} of the attack on a hardware board (\cref{subsec:zedboard_eval}) running a real-time operating system.

In summary, the main contributions of this paper are:
\begin{enumerate}
\itemsep 0pt
\item a reconnaissance attack scheme aimed at {\em extracting the schedule} of fixed-priority, preemptive
	hard real-time systems -- this involves development of algorithms
	capable of {\em reconstructing} the schedule of the system (\cref{sec::approach});
%
\item introducing new performance metrics for the precision of the schedule inference (\cref{subsec::metrics}) and evaluating the attack scheme using exhaustive simulations (\cref{subsec:simulation_eval});

\item implementing the attack scheme using a realistic embedded platform (\cref{subsec:zedboard_eval});

\item demonstrating one use case where a reconstructed schedule is used to launch other attacks, \eg a cache-based side-channel attack (\cref{sub:eval_cache_attack}).
\end{enumerate}

\section{System and Attack Model}
\label{sec::model}

In this section we first introduce the real-time system model used in this paper. 
Next, we discuss the attack model and the information needed for successfully launching the proposed attack. 
A demonstrative system and attack scenarios are given in the end.
For ease of reference, the notation defined for the \ATTCKNF algorithms is summarized in Table~\ref{tab:notation}.

\newcommand{\tabspace}{\hspace{0.3cm}}
\begin{table}[h]\footnotesize
\caption{Glossary of the \ATTCKNF notation.}
\label{tab:notation}
\centering
\vspace{-0.5\baselineskip}
\begin{tabular}{|l|l|}
\hline
\multicolumn{1}{|c|}{Symbol} & \multicolumn{1}{c|}{Definition} \\ \hline
$W$ & busy interval set $\{\omega_1, ..., \omega_m\}$\\ \hline
\hspace{0.5mm} $\omega_k$ & a busy interval \\ \hline
\hspace{2.5mm} $s_k, l_k$ & start time and length of $\omega_k$ \\ \hline
\hspace{2.5mm} $J_k$ & a job set enclosed in $\omega_k$ \\ \hline
\hspace{2.5mm} $\NKSET$ & a subset of $J_k$ consisting of jobs of $\tau_i$ \\ \hline
\hspace{4.5mm} $j^{\tau_i}_{k,h}$ & the $h^{th}$ job in $\NKSET$ \\ \hline
\hspace{6.5mm} $\ATME, \STME$ & arrival time and start time of $j^{\tau_i}_{k,h}$ \\ \hline
$R$ & \ATTCKNF $R(\Gamma, W)={J}$ where ${J}=\bigcup^{m}_{k=1} J_k$ \\ \hline
\hspace{0.5mm} $\overline{a_i}$ & an arrival window inferred by $R$ (defined in \cref{approach:arrival_windows}) \\ \hline
\end{tabular} 
\vspace{-0.15in}
\end{table}

\subsection{Real-time System and Schedule Model}
\label{subsec:sysmodel}
In this paper, a fixed-priority, preemptive hard real-time system 
is considered. Such a system contains a task set $\Gamma=\{\tau_1, ..., \tau_n\}$ consisting of $n$ periodic real-time tasks with hard deadlines.
Each task $\tau_i$, $1\le i\le n$, is characterized by $p_i$ (period), $wcet_i$ (worst case execution time), 
$c_i$ (actual execution time)\footnote{We assume that the actual execution time at run time is $c_i + \gamma_{i,h}$, where $\gamma_{i,h}$ represents the variation (which includes jitters) at the $h^{th}$ job. However, for simplicity, we will use $c_i$ without variation $\gamma_{i,h}$ to illustrate the proposed algorithm in \cref{sec::approach}.}
, $d_i$ (deadline), $a_i$ (initial arrival time, known as task offset) and $pri_i$ (priority)
based on \cite{LiuLayland1973}.
%
%
%
We model the schedule of a periodic real-time system within one hyper-period as a set of intervals $HP$ that consists of two subsets: \ci a subset of busy intervals $W$ where tasks are under execution \cite{baker2005analysis} and \cii a subset of idle intervals $ID$ where there is no task being executed. 
Specifically, the busy interval set $W$ containing $m$ busy intervals is defined as $W=\{\omega_1, ..., \omega_{m}\}$. Each busy interval $\omega_k$, $1 \leq k \leq m$, is characterized by a tuple $(s_k, l_k, J_k)$, in which $s_k$ and $l_k$ are the start time and the length of $\omega_k$, and $J_k$ denotes the job set (instances of tasks) comprising $\omega_k$. 
We further define the notation $J^{\tau_i}_k$ to be the subset of $J_k$ for $\tau_i$ in $\omega_k$, thus $J_k=\bigcup^{n}_{i=1}J^{\tau_i}_k$. The $h^{th}$ job in $J^{\tau_i}_k$ is then denoted by $j^{\tau_i}_{k,h}$, that is defined by $(\ATME, \STME)$ where $\ATME$ is the arrival time when the job is being scheduled and $\STME$ is the start time when the job begins to execute. 
$\STME$ is greater than or equal to $\ATME$ when $j^{\tau_i}_{k,h}$ has to wait for higher-priority tasks executing before it to finish. 

\subsection{Attack and Adversary Model}
\label{subsec:attackmodel}
The goal of the adversary in our model is to steal information while remaining undetected rather than to disable or disrupt the real-time system. 
Therefore, reconnaissance to obtain knowledge about the system's operation becomes even more important than in enterprise settings as attack attempts or unsuccessful attacks can more easily be detected due to the deterministic and predictable nature of the system. 

We assume that the adversary has a foothold on the real-time system (the attacker is able to run one or more tasks in the system) and has knowledge of the system task set $\Gamma$. 
For example, in a multi-vendor development model~\cite{embeddedsecurity:mohan2015}, an attacker may obtain task set information offline by compromising one of the vendors with weak defenses. Similarly, execution time distributions, that provide the execution times (and estimated variations), obtained during WCET analysis -- an essential process in developing a real-time system \cite{hansen2009statistical, 
burns2000predicting,
cazorla2013proartis} may also be obtained offline. 
A demonstrative avionics system that exemplifies this model along with an attack scenario are introduced in \cref{subsec:example_case_study}.
These assumptions are not unreasonable as demonstrated by Stuxnet \cite{chen2011lessons}, a multi-stage attack in which the attackers compromised many other systems before reaching their target system, and in which the attackers stayed undetected for months conducting reconnaissance on the target system before they launched their attacks.

However, we also assume that the adversary has no knowledge of the initial arrival time $a_i$ for each task. Since initial arrival times of tasks are subject to the system's operating condition, and the adversary's attack code might have started sometime after system starts, the values of $a_i$ are unknown {\it a priori} to the adversary.
Adversary's goal is to use the known data (\ie task sets) to learn the system's run-time behavior and operations without being detected. Specifically, \emph{our focus will be on reconstructing the system's task schedule}.
This information can be used to increase the success rate of other attacks as exemplified by the side-channel attack discussed in \cref{subsec:example_case_study}.

\noindent
\textbf{Observer Task:} 
In preemptive RTS, low priority tasks are particularly useful in monitoring the system's behavior, \eg for measuring the execution time of the task that preempts it. Thus, the lowest priority task 
has the ability to monitor the schedule of a system by observing its busy intervals $W=\{\omega_1,...,\omega_m\}$. 
We define such a task 
as the \emph{``\OBSERVER''}. 
%
In this paper, we assume that an adversary can gain access to an \OBSERVER~\ci either by utilizing an existing lowest priority task or \cii by inserting a new task with the lowest priority and infinite period (hence being continuously active).
%
Note 
the observed intervals are actually a composite of multiple tasks thus making it difficult for direct use without further analysis. 
Hence, an algorithm to extract $J_k$ from $W$ is required.

\subsection{Demonstrative System and Attack}
\noindent
\textbf{Avionics Demonstrator:} 
\label{subsec:example_case_study}
To motivate 
our research, we use an example of the Electronic Control Unit (ECU) for an avionics system as depicted in Figure~\ref{fig:sys_example}. This
system runs many of the same types of tasks that could be expected in an Unmanned Aerial Vehicle (UAV) surveillance system~\cite{embeddedsecurity:mohan2015, puri2005survey}. The ECU communicates locally with the inertial sensors, GPS 
system and actuators (``UAV'' in the figure), as well as a camera subsystem. The ECU also uses off-board communication to exchange information with a base station. 
We assume that three parties are involved in building the ECU system: Vendor 1, Vendor 2 and the Integrator. Each party is responsible for a different ECU subsystem and real-time tasks (See Figure~\ref{fig:sys_example}). 

\begin{figure}[h]
\vspace{-1\baselineskip}
\centering
\includegraphics[width=0.85\columnwidth]{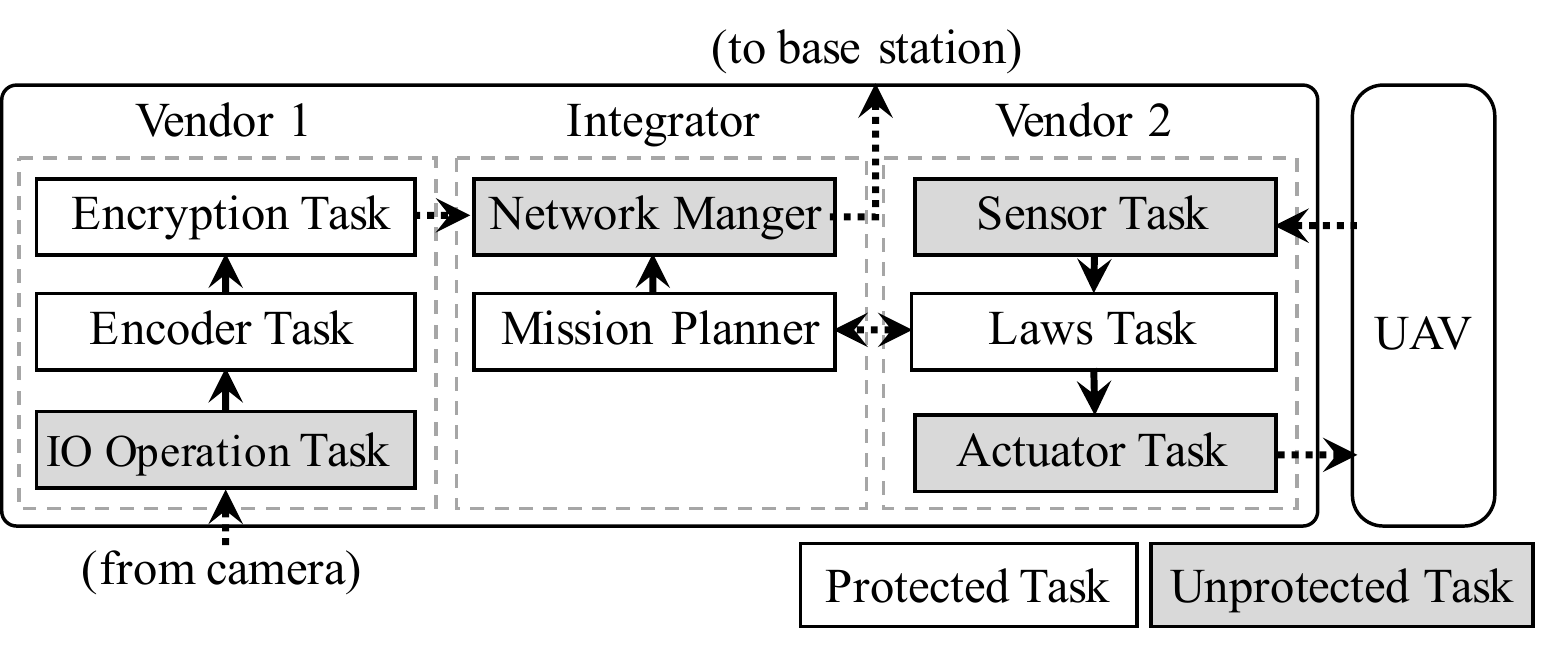}
\vspace{-0.5\baselineskip}
\caption{High-level design of an Unmanned Aerial Vehicle.}
\label{fig:sys_example}
\vspace{-0.5\baselineskip}
\end{figure}

\noindent
\textbf{Targeted Attacks:}
An attacker who accomplishes the reconstruction of task schedule is able to pinpoint the arrival instant and the start instant of any task on a victim system.
Here we introduce an attack that spies upon execution behavior of the target task
with the knowledge of the precise task schedule.

Vendor 2 who manages sensor readings has access to the UAV coordinates. This vendor, as an adversary, would like to identify the exact locations of high-interest targets where the surveillance camera is enabled.
The attacker can achieve that by 
learning the behavior (\eg memory usage, execution modes) of the Encoder Task.
For instance, he can launch a storage-channel-based timing attack 
to observe the resource usage of the Encoder Task
and learn whether the task is processing a large amount of data.
This attack requires precise knowledge of system operation as a careless attack can result in \ci increased noise and less precision in results, or worse \cii a real-time task missing its deadline and thus lead to detection. In this paper, we implement a cache-based side-channel attack on a hardware board to demonstrate the importance and usability of the task schedule inference (See~\cref{sub:eval_cache_attack}).

\section{\ATTCKNF}
\label{sec::approach}

To effectively extract the system's task schedule, we propose 
``\ATTCK'', {\em a novel algorithm that uses the predictable, periodic, nature of hard real-time systems to reconstruct the schedule}. 
It can be represented as a function
$R(\Gamma, W)={J}$,
where $R$, \ATTCKNF algorithm, takes the target system's task set $\Gamma$ and busy interval set $W$ as inputs and produces the inferred job set ${J}=\bigcup^{m}_{k=1} J_k$.
Here, the task set $\Gamma$ is known except for the arrival time attribute $a_i$ 
and the busy interval set $W$ is obtained from the \OBSERVER controlled by the attacker while the details of job set $J_k$ stay unknown. Finally, ${J}$
allows the attacker to pinpoint the possible start time of any particular victim task. 
The sequence of iterative steps in \ATTCKNF are summarized in Figure~\ref{fig:attack_flow_chart}.
A detailed illustration is given next.
\begin{figure}[h]
\vspace{-1\baselineskip}
\centering
\includegraphics[width=0.9\columnwidth]{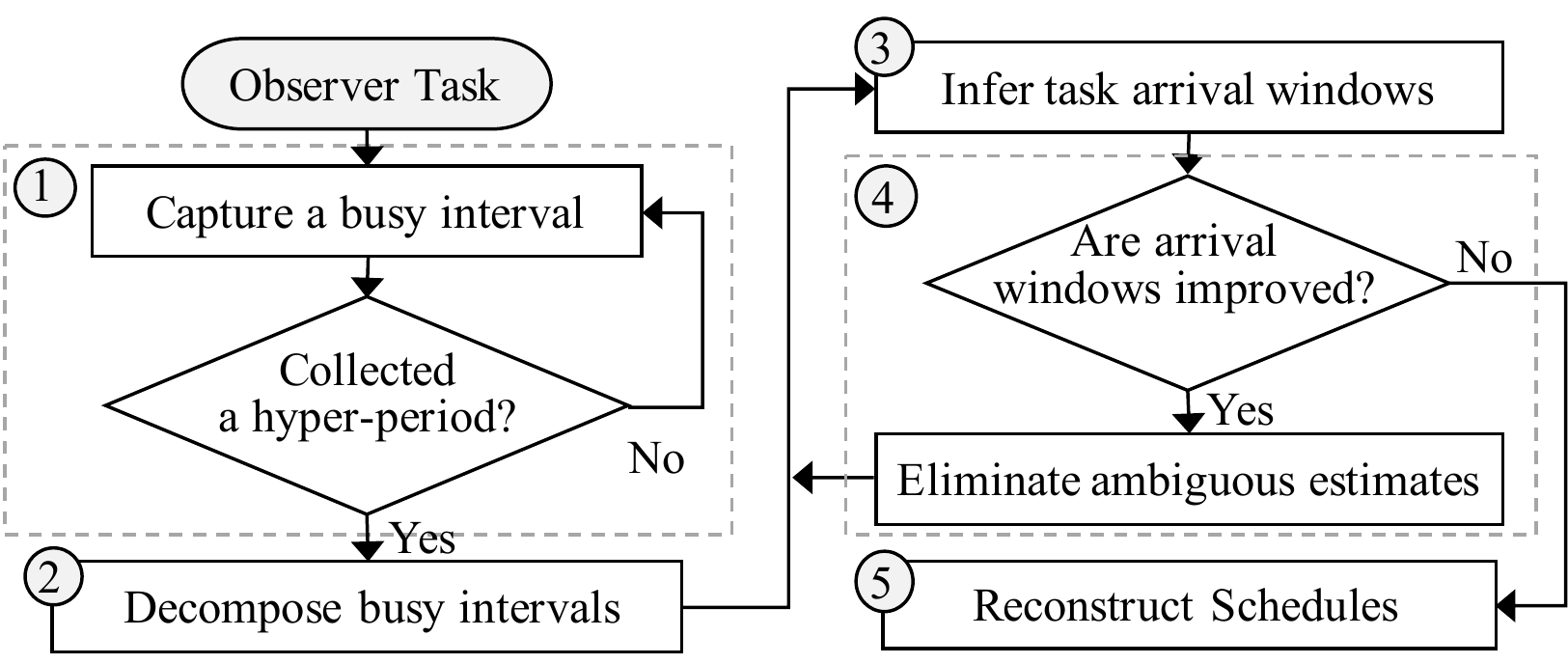}
\caption{Flow of the $\ATTCKNF$ algorithm. }
\label{fig:attack_flow_chart}
\vspace{-0.25in}
\end{figure}



\subsubsection{Capturing Busy Intervals}
\label{approach:capture_busy_intervals}
Recall from \cref{subsec:attackmodel} that the adversary utilizes his \OBSERVER to capture busy intervals.
To do so, {\em the \OBSERVER checks whether it has been preempted}. 
When a preemption is detected, a busy interval $\omega_k$ is found. The duration between preemptions is $l_k$ and the start time of this busy interval is $s_k$. 
The detailed mechanism for capturing a busy interval is presented in \cref{implementation}.
This process of capturing one busy interval repeats until the \OBSERVER collects all busy intervals in one hyper-period (since the schedule of arrivals repeats every hyper-period).
Eventually, the adversary can get the complete busy interval set $W=\{\omega_1, ..., \omega_m\}$. 
Note that \ATTCKNF can also work with busy intervals observed from fewer or more than one hyper-period as discussed in~\cref{subsec::simulation_results}

\subsubsection{Decomposing Busy Intervals}
\label{approach:estimate_of_nk}
The goal of this step is to estimate possible job compositions of ${J_k}$ in each busy interval $\omega_k$. 
The duration of $\omega_k$ can be calculated by
%
\vspace{-0.5\baselineskip}
\small
\begin{equation}\label{eqn:biLength}
l_k=\sum_{i=1}^{n} (\NK \cdot c_i)
\vspace{-0.3\baselineskip}
\end{equation}
\normalsize
\noindent
where $l_k$ and $c_i$ are known to the adversary. Therefore, the problem here can be further refined as finding the number of jobs $\NK$ for each task $\tau_i$ arriving in $\omega_k$.

A task $\tau_i$ may either contribute nothing or contribute one or more jobs to $\omega_k$, \ie $0 \leq \NK$. But, given the duration $l_k$ of the busy interval, the number of jobs for a task $\tau_i$ depends on its period $p_i$ and execution time $c_i$. Therefore, knowing $p_i$ and $c_i$ of every task, we can reduce the number of possibilities for the combination of $\NKSET$. 
We use Theorem~\ref{thm:estimate_nk} described below to either estimate the exact value of $\NK$ or, in the worst case, reduce the possible candidates for $\NK$ to only 2 values.

\begin{theorem}
\label{thm:estimate_nk}
For the task $\tau_i$, given values of $p_i,c_i,\CK$:\\
\noindent
(i) If $l_k$ satisfies 
$(N \cdot p_i-c_i)^+ \le \CK < N \cdot p_i+c_i$,
then $\omega_k$ can only contain $N$ jobs for $\tau_i$.

\noindent
(ii) If $\CK$ satisfies
$N \cdot p_i+c_i \le \CK < (N+1) \cdot p_i-c_i$,
then $\omega_k$ can only contain $N$ or $N+1$ jobs for $\tau_i$. 
\end{theorem}

We will use the following two lemmas to prove Theorem~\ref{thm:estimate_nk}.
Consider $\sigma^{\tau_i}_{0}$, $\sigma^{\tau_i}_{1}$, $\sigma^{\tau_i}_{2}$, $\cdots$ as the arrival times of task $\tau_i$ and define $\mathcal{S}=\{\sigma^{\tau_i}_{0}, \sigma^{\tau_i}_{1}, \sigma^{\tau_i}_{2}, \cdots\}$. 
\begin{lemma}
\label{lem:0}
\vspace{-0.2\baselineskip}
 A busy interval contains the $h^{th}$ job of $\tau_i$ if and only if it contains $\sigma^{\tau_i}_{h}$. 
\end{lemma}
\begin{proof}
The $h^{th}$ job of $\tau_i$ will be released at time $\sigma^{\tau_i}_{h}$ if there are no higher priority tasks running at that time, or, will be released immediately after the end of the higher priority tasks. 
In both cases, the system is busy from time $\sigma^{\tau_i}_{h}$ to at least the finishing time of the $h^{th}$ job of $\tau_i$. Therefore, a busy interval which contains $\sigma^{\tau_i}_{h}$ has the $h^{th}$ job of $\tau_i$, and a busy interval that contains the $h^{th}$ job of $\tau_i$, should have started at $\sigma^{\tau_i}_{h}$ or at a time before $\sigma^{\tau_i}_{h}$. (Note that the end point of a busy interval cannot belong to $\mathcal{S}$).
\end{proof}

\begin{lemma}
\label{lem:1}
If $l_k$ satisfies
$Np_i<l_k<(N+1)p_i,N=0,1,2,...$,
then task $\tau_i$ can only have arrived $N$ or $N+1$ times during the busy interval $w_k$.
\end{lemma}
\begin{proof}
If $Np_i<l_k<(N+1)p_i$, then $w_k$ contains $N$ or $N+1$ points of $\mathcal{S}$. Therefore, by Lemma \ref{lem:0}, task $\tau_i$ can only arrive $N$ or $N+1$ times during the busy interval $w_k$.
\end{proof}

\begin{proof}[Proof of Theorem 1]
(i) If $Np_i-c_i\le l_k<Np_i$, then the busy interval cannot contain $N-1$ points from  $\mathcal{S}$, otherwise, task $\tau_i$ should have finished in a time interval less than $c_i$. Therefore, it exactly contains $N$ points from  $\mathcal{S}$.\\
If $Np_i\le l_k<Np_i+c_i$, the start point of the busy interval cannot belong to $\mathcal{S}$ (otherwise, the $l_k$ should be at least $Np_i+c_i$), therefore, it exactly contains $N$ points from $\mathcal{S}$.
Therefore, by Lemma \ref{lem:0}, in both cases, task $\tau_i$ can only have arrived $N$ times during the busy interval.\\
(ii) This part follows from Lemma \ref{lem:1} immediately.
\end{proof}

Note that Theorem~\ref{thm:estimate_nk} presented here does not take execution time variation in
$c_i$ into account for ease of exposition. 
Nevertheless, the theorem can be easily extended to cover variation by adding tolerance bounds for each variable.
%
Further, our evaluation (\cref{sec::eval}) takes execution time variation into account.   
%


Based on Theorem~\ref{thm:estimate_nk}, each task $\tau_i$, in worst case, may have two candidate $\NK$ values for a busy interval: $N$ and $N+1$.
Therefore, computing the right term in Equation (\ref{eqn:biLength}) with all possible $\NK$ values, there will be at most $2^n$ combinations. Nonetheless, only those combinations that yield exactly $l_k$ and satisfy Equation (\ref{eqn:biLength}) are shortlisted for further consideration.


\begin{example}
\label{ex:complete_estimate_nk_1}
Consider a task set $\Gamma=\{\tau_1, \tau_2, \tau_3\}$ and 4 captured busy intervals in a hyper-period, i.e., $LCM(5, 6, 10)=30$, as follows:
\begin{center}\footnotesize
\vspace{-0.5\baselineskip}
\begin{tabular}{|c||c|c|}
\hline 
 & $p_i$ & $c_i$ \\ 
\hline 
$\tau_1$ & 5 & 1 \\ \hline 
$\tau_2$ & 6 & 2 \\ \hline 
$\tau_3$ & 10 & 2 \\ \hline 
\multicolumn{3}{r}{} \\
\end{tabular} 
\quad
\begin{tabular}{|c||c|c|}
\hline 
$W$ & $l_k$ & Timestamp \\ 
\hline 
$\omega_1$ & 8 & [0,8] \\ 
\hline 
$\omega_2$ & 6 & [10,16] \\ 
\hline 
$\omega_3$ & 5 & [18,23] \\ 
\hline 
$\omega_4$ & 3 & [24,27] \\ 
\hline 
\end{tabular} 
\end{center}


\noindent Taking busy interval $\omega_4$ as an example, by applying Theorem~\ref{thm:estimate_nk} to $\omega_4$, we can compute $\vert J^{\tau_1}_4 \vert$, $\vert J^{\tau_2}_4 \vert$ and $\vert J^{\tau_3}_4 \vert$ values presented in the left table below.
And then we use Equation~(\ref{eqn:biLength})
to find possible $J_4$ combinations that can lead to the given busy interval duration $l_4=3$ as shown in the table to the right in the following.
As a result, most combinations are eliminated except the two, 
$\{1, 0, 1\}$ and $\{1, 1, 0\}$, that can lead to a busy interval length of $l_4=3$. In this case, both inferences will be retained for further processing.

\begin{center}\footnotesize
\vspace{-0.5\baselineskip}
\begin{tabular}{|c||c|}
\hline 
 & $\vert J^{\tau_i}_4 \vert$ \\ 
\hline 
$\tau_1$ & 0 or 1 \\ \hline 
$\tau_2$ & 0 or 1 \\ \hline 
$\tau_3$ & 0 or 1 \\ \hline 
\multicolumn{2}{r}{} \\
\multicolumn{2}{r}{} \\
\multicolumn{2}{r}{} \\
\multicolumn{2}{r}{} \\
\multicolumn{2}{r}{} \\
\multicolumn{2}{r}{} \\
\end{tabular} 
\quad
\begin{tabular}{|c c c||c|c}
\cline{1-4}
$\vert J^{\tau_1}_4 \vert$ & $\vert J^{\tau_2}_4 \vert$ & $\vert J^{\tau_3}_4 \vert$ & $l_4$ &  \\ \cline{1-4} \cline{1-4}
0 & 0 & 0 & 0 & \\ \cline{1-4} 
0 & 0 & 1 & 2 & \\ \cline{1-4} 
0 & 1 & 0 & 2 & \\ \cline{1-4} 
0 & 1 & 1 & 4 & \\ \cline{1-4} 
1 & 0 & 0 & 1 & \\ \cline{1-4} 
1 & 0 & 1 & 3 & $\surd$ \\ \cline{1-4} 
1 & 1 & 0 & 3 & $\surd$ \\ \cline{1-4} 
1 & 1 & 1 & 5 & \\ \cline{1-4} 
\multicolumn{4}{r}{$^*$desired $l_4=3$} & \multicolumn{1}{c}{} \\
\end{tabular} 
\end{center}
\vspace{-1\baselineskip}

\end{example}

\subsubsection{Inferring Task Arrival Windows}
\label{approach:arrival_windows}
Here, we estimate possible arrival times for each task. 
Rather than directly inferring a point for the arrival time $a_i$, we compute a possible window for that arrival, that we call \emph{arrival window}, defined as $\overline{a_i}$. 
It is done by computing each job's arrival window, defined as $\ATMEW$, in every busy interval and merging the results in the end for each task.

To identify potential arrival windows $\ATMEW$ for $\tau_i$ in $\omega_k$, we partition the busy interval $\omega_k = [\alpha, \beta]$ into the following three types of segments for each task $\tau_i$:\\

\vspace{-0.5\baselineskip}
\begin{tabularx}{\columnwidth}{l l X}
$-$ & \emph{0-interval}: & There is no arrival. \\ 
$-$ & \emph{1-interval}: & There exists exactly 1 arrival. \\ 
$-$ & \emph{0-1-interval}: & There may exist 0 or 1 arrivals. \\ 
\end{tabularx} 



%
%

\noindent The partitioning of the busy interval is done using the following theorem. 

\noindent
\begin{theorem}
\label{thm:arrival_windows}
Considering a task $\tau_i$ and a busy interval $w_k$ that has start time $\alpha$ and end time $\beta$. The partitioning of the busy interval is done by using the following equations:\\
(i) If $\tau_i$ has arrived exactly $N$ times:\\
\vspace{-0.5\baselineskip}
If $\displaystyle N=\Big\lceil \frac{l_k}{p_i} \Big\rceil$, the following segments are \emph{1-interval}:
\small
\begin{equation}
\label{eqn:N-1-interval}
\ATMEW = [\alpha+(h-1)p_i, 
 \beta-(N-h)p_i-c_i] \hspace{0.6cm} 1\le h \le N
\end{equation}
\normalsize

Else, the following segments are \emph{1-interval}:
\small
\begin{equation}
\label{eqn:N-1-interval-2}
\ATMEW = [\beta+(N+1-h)p_i, 
 \alpha+hp_i-c_i] \hspace{1cm} 1\le h \le N
\end{equation}
\normalsize

\noindent
(ii) If $\tau_i$ has potentially arrived either $N$ or  $N+1$ times: \\
\vspace{-0.2\baselineskip}
the following segments are \emph{1-interval}:
\small
\begin{equation}
\label{eqn:N-1-interval-3}
\ATMEW = [\alpha+(h-1)p_i, \alpha+hp_i-c_i]\hspace{1cm} 1\le h \le N
\end{equation}
\normalsize

and the following segments are \emph{0-1-interval}:
\small
\begin{equation}
\label{eqn:N-1-interval-4}
\ATMEW = [\alpha+(h-1)p_i, \beta-c_i]\hspace{1cm}  h=N+1
\end{equation}
\normalsize

where $\ATMEW$ is the $h^{th}$ arrival window for $\tau_i$ in $\omega_k$.
In both cases, the remainder of the busy interval is \emph{0-interval}.
\end{theorem}

\begin{proof}
(i) From Theorem~\ref{thm:estimate_nk}, If $\NK=N$, then
$(N-1)p_i+c_i\le l_k\le(N+1)p_i-c_i$.
We partition the busy interval as follows:
$[\alpha,\beta]=[\alpha,\alpha+p_i]\cup[\alpha+p_i,\alpha+2p_i]\cup\cdots\cup[\alpha+hp_i,\beta]$.
If $(N-1)p_i+c_i\le l_k\le Np_i$, or equivalently, $N=\lceil \frac{l_k}{p_i} \rceil$, then $h=N-1$ and there should be an arrival in each interval of the partition above. Also, there cannot be an arrival in $[\beta-c_i,\beta]$, otherwise, the busy interval cannot terminate at $b$. Therefore, there should be an arrival in $[\alpha+(N-1)p_i,\beta-c_i]$.\\
By the periodicity assumption, task $\tau_i$ arrives every $p_i$ seconds, i.e., if we have an arrival at time $t$, then, $t\pm p_i$ is also arrival times. So, shifting $[\alpha+(N-1)p_i,\beta-c_i]$ by integer multiples of $p_i$ to the left and taking its intersection with other intervals of the partition, we get Equation~(\ref{eqn:N-1-interval}).\\
If $Np_i\le l_k\le (N+1)p_i-c_i$, then $h=N$ and there should not be any arrivals in the last interval of the partition above. Hence, because of the periodicity, the last arrival should be in the interval $[\beta-p_i, \alpha+Np_i]$. 
But, if there is an arrival in interval $[\alpha+Np_i-c_i, \alpha+Np_i]$, then there should be an arrival in the interval $[\alpha-c_i, \alpha]$. Therefore, the busy interval cannot start at $\alpha$. This implies that the last arrival should be in the interval $[\beta-p_i, \alpha+Np_i-c_i]$. 
Finally, because of the periodicity, by shifting $[\beta-p_i, \alpha+Np_i-c_i]$ by integer multiples of $p_i$ to the left and taking its intersection with other intervals of the partition, we get Equation~(\ref{eqn:N-1-interval-2}).\\
\noindent
(ii) Using Theorem 1 again, if $\NK=N\text{ or }N+1$, then
$Np_i+c_i\le l_k\le (N+1)p_i-c_i$.
Therefore, in the partition, $h=N$, and we cannot say anything about the last interval of the partition. So, we mark it as a 0-1-interval, with the consideration that similar to part (i), there cannot be an arrival in $[\beta-c_i,\beta]$. This gives us Equation~(\ref{eqn:N-1-interval-4}). Also, there should be an arrival in all other intervals of the partition, with the consideration that similar to part (i), there cannot be an arrival in $[\alpha+hp_i-c_i, \alpha+hp_i]$, for $h=1,..,N$. This gives us Equation~(\ref{eqn:N-1-interval-3}).
\end{proof}


Figure~\ref{fig:arrv_seg} shows the use of Theorem~\ref{thm:arrival_windows}. Part (a) depicts a case in which $\NK=3$ and $\beta-(3-h)p_i-c_i<\alpha+hp_i-c_i$, for $h=1,2,3$. Note that 1-intervals for $\tau_i$ in a busy interval should be repeated every $p_i$.
Part (b) depicts a case in which $\NK=2$ or $3$ and we are not able to determine whether the last interval contains an arrival; hence, it will be a \emph{0-1-interval}.

\begin{figure}[h]
\vspace{-0.5\baselineskip}
    \centering
    \begin{subfigure}[t]{0.49\columnwidth}
        \centering
        \includegraphics[width=0.99\columnwidth]{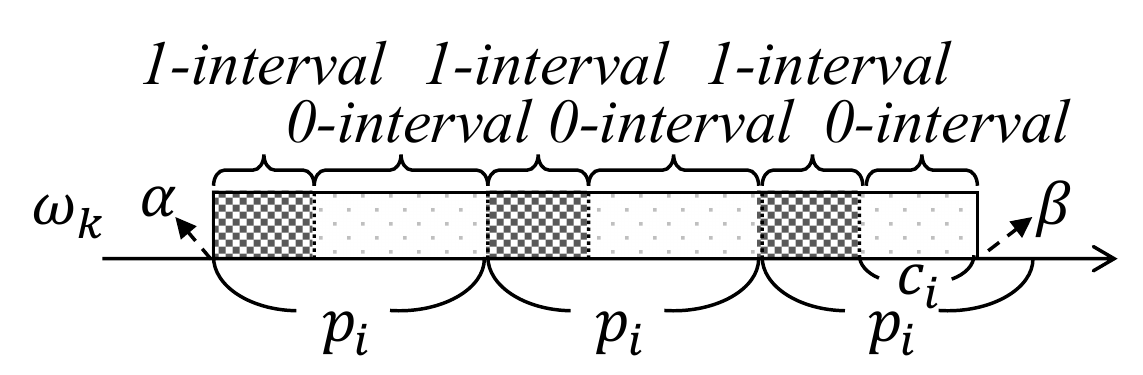}
        \vspace{-1.5\baselineskip}
        \caption{}
	\label{fig:arrv_seg_a}
    \end{subfigure}
    \begin{subfigure}[t]{0.49\columnwidth}
        \centering
       \includegraphics[width=0.99\columnwidth]{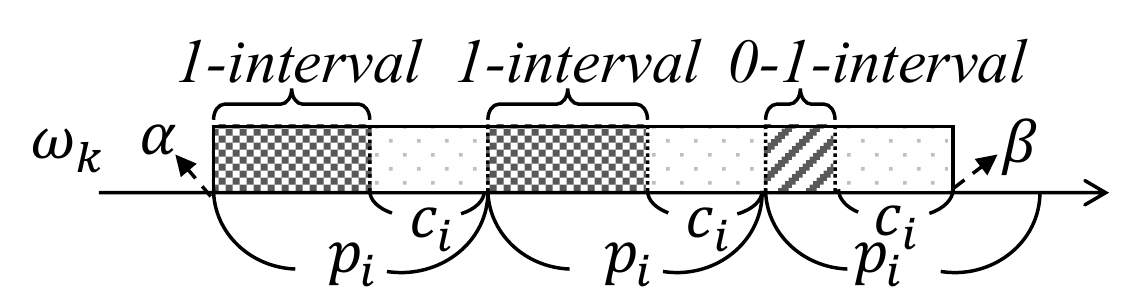}
        \vspace{-1.5\baselineskip}
        \caption{}
	\label{fig:arrv_seg_b}
    \end{subfigure}
    \vspace{-0.5\baselineskip}
    \caption{Estimate of arrival time locations.}
    \label{fig:arrv_seg}
    \vspace{-0.5\baselineskip}
\end{figure}


We then use the above segments to find the initial arrival window $\overline{a_i}$ for $\tau_i$. 
Because of periodicity,
a task must arrive exactly once in each time interval of length equal to its period. 
Moreover, without considering jitters, the relative arrival time in each period should be consistent. 
Thus, if we divide one hyper-period into intervals of length $p_i$ and overlap them together, we can obtain a distribution of possible arrival locations of $\tau_i$ in its period $p_i$.
Here, we take the overlapped segments that have the highest occurrence probability as the arrival window $\overline{a_i}$. 
We will use the following example to demonstrate how this works.

\begin{figure*}[t]
\centering
    \begin{subfigure}[t]{0.27\linewidth}
        \centering
        \includegraphics[width=0.9\columnwidth]{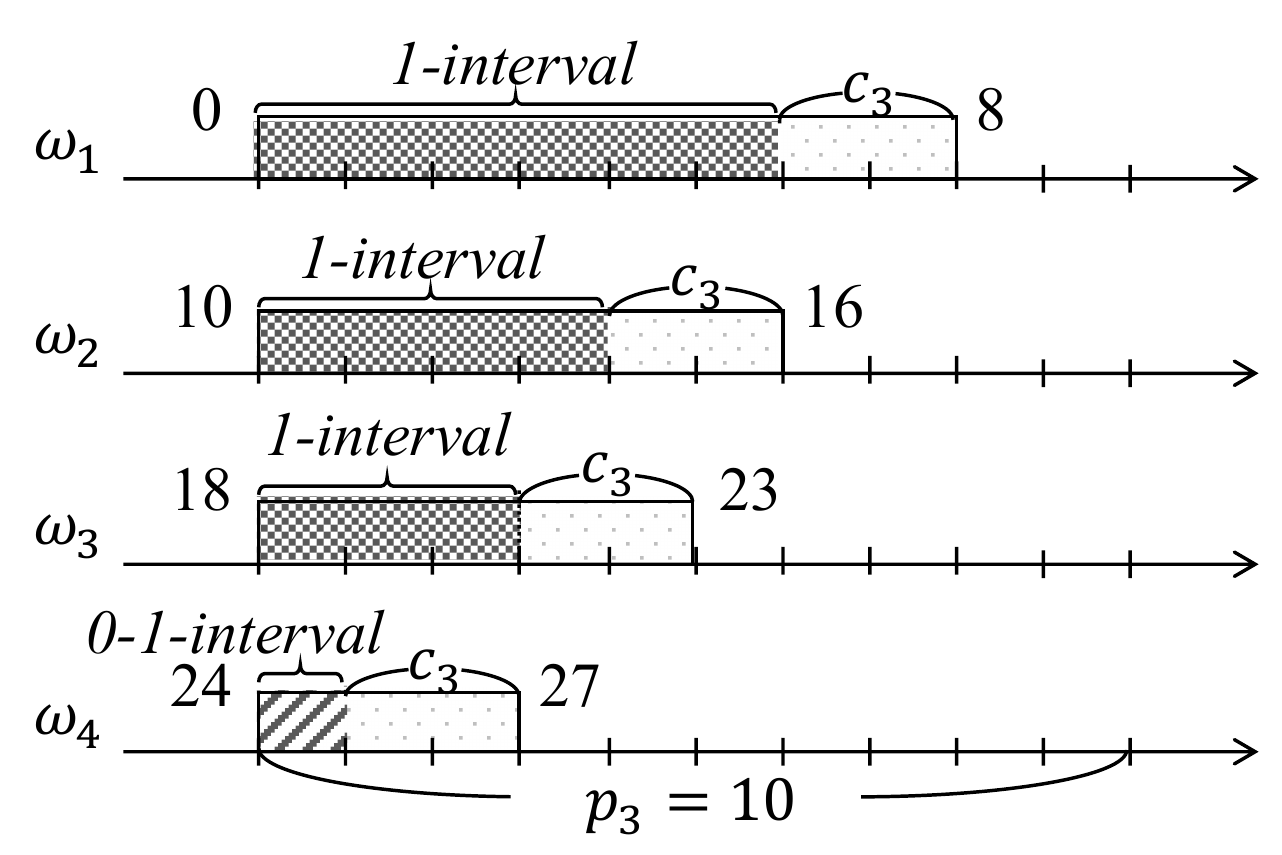}
\vspace{-0.35\baselineskip}
\caption{Partition busy intervals for $\tau_3$ based on Theorem~\ref{thm:arrival_windows}.}
    \end{subfigure}%
        ~
    \begin{subfigure}[t]{0.25\linewidth}
        \centering     
        \includegraphics[width=0.90\columnwidth]{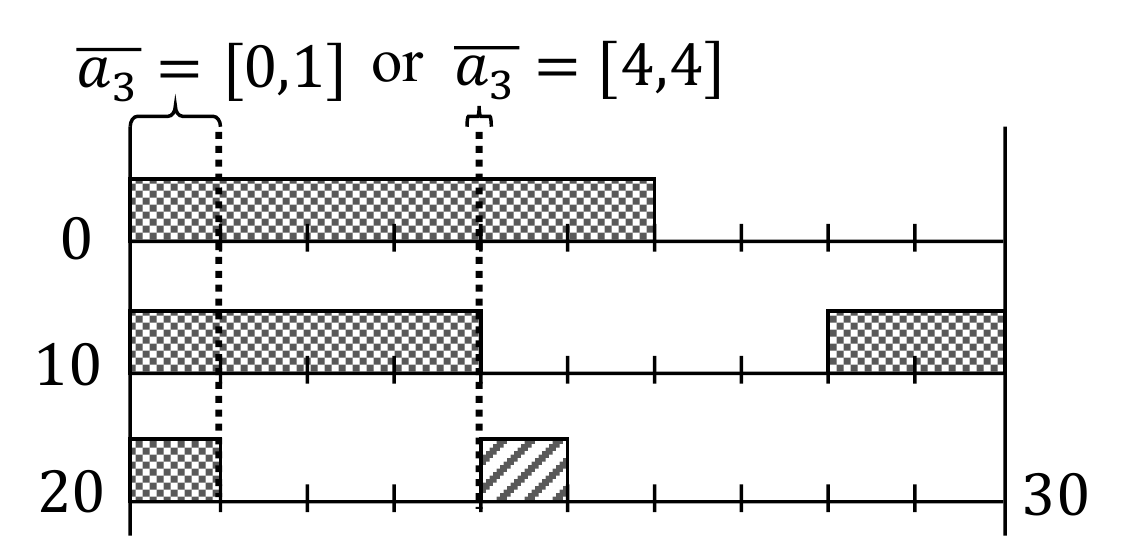}
        \vspace{-0.35\baselineskip}
        \caption{Arrival windows of $\tau_3$}
    \end{subfigure}%
        ~ 
    \begin{subfigure}[t]{0.17\linewidth}
        \centering     
        \includegraphics[width=0.90\columnwidth]{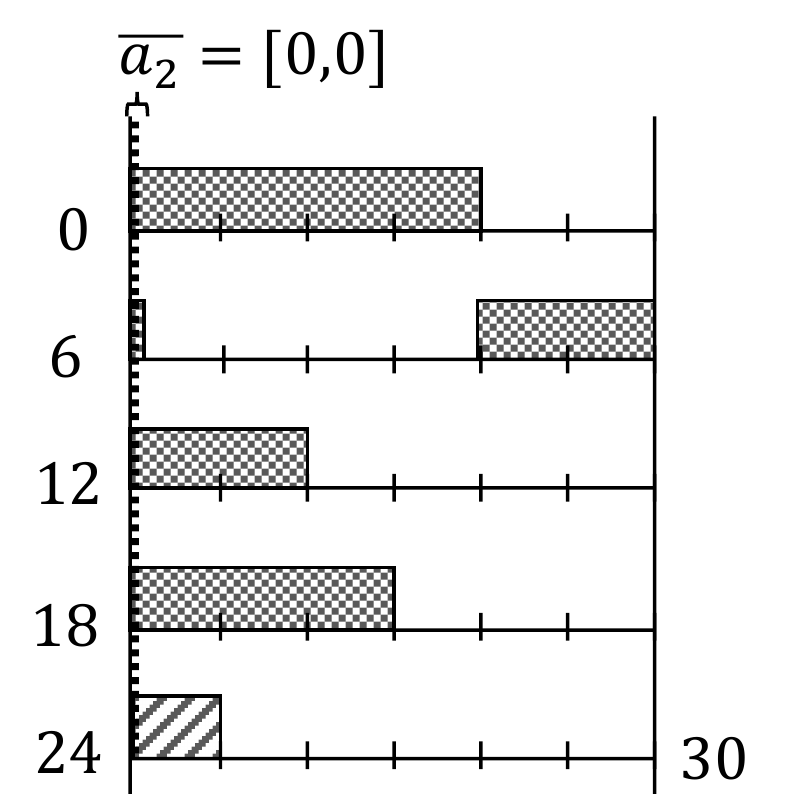}
        \vspace{-0.35\baselineskip}
        \caption{Arrival windows of $\tau_2$}
    \end{subfigure}
    ~
     \begin{subfigure}[t]{0.25\linewidth}
        \centering     
    \includegraphics[width=0.9\columnwidth]{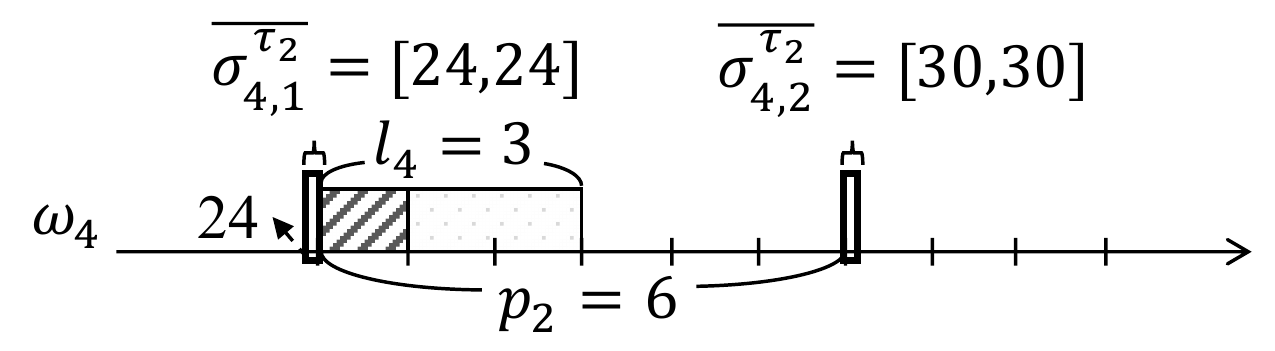}
        \vspace{-0.35\baselineskip}
        \label{fig:example_remove_mismatched}
        \caption{Validate $\vert J^{\tau_2}_4 \vert$}
    \end{subfigure}   

  \vspace{-0.5\baselineskip}
  \caption{Example~\ref{ex:complete_infer_arrival_windows} and \ref{ex:complete_remove_mismatched_nk}, compute possible arrival windows and validate inferences.}
    \label{fig:example_intersect_arrival_window}
    \vspace{-1.5\baselineskip}

\end{figure*}

\begin{example}
\label{ex:complete_infer_arrival_windows}
Here we consider task $\tau_3$ from Example~\ref{ex:complete_estimate_nk_1} to demonstrate how the arrival window $\overline{a_3}$ is derived.
By following Theorem~\ref{thm:arrival_windows}, four busy intervals can be partitioned into segments
as shown in Figure~\ref{fig:example_intersect_arrival_window}(a), and the corresponding overlaps are plotted in Figure~\ref{fig:example_intersect_arrival_window}(b).
In this case, there are two 
segments that have the highest overlap count, 3,
that represent two possible arrival windows $\overline{a_3}=[0,1]$ or $\overline{a_3}=[4,4]$ for task $\tau_3$.
Similarly, arrival windows for $\tau_1$ and $\tau_2$ can be computed as $\overline{a_1}=[0,0]$ and $\overline{a_2}=[0,0]$, respectively. Figure~\ref{fig:example_intersect_arrival_window}(c) shows the arrival window for $\tau_2$.

\end{example}

\begin{figure}[t]
    \centering
    \begin{subfigure}[t]{0.49\columnwidth}
        \centering
        \includegraphics[width=0.99\columnwidth]{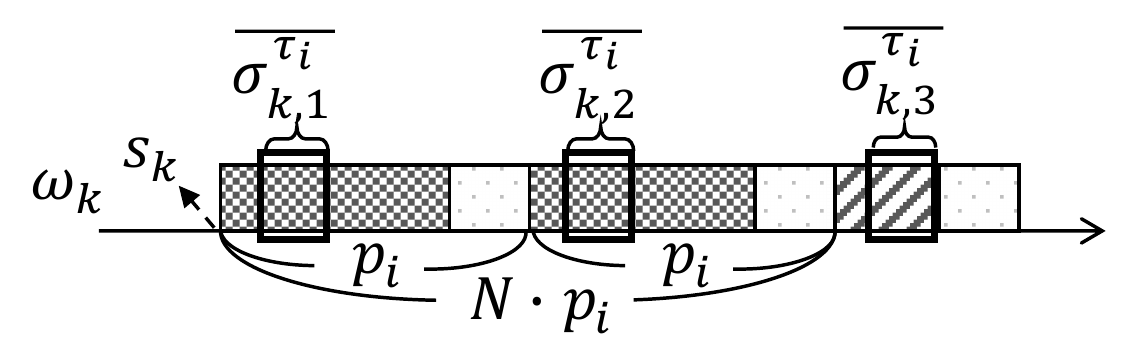}
        \vspace{-1.5\baselineskip}
        \caption{$\NK=N+1$ is clarified.}
    \end{subfigure}%
    \begin{subfigure}[t]{0.49\columnwidth}
        \centering
        \includegraphics[width=0.99\columnwidth]{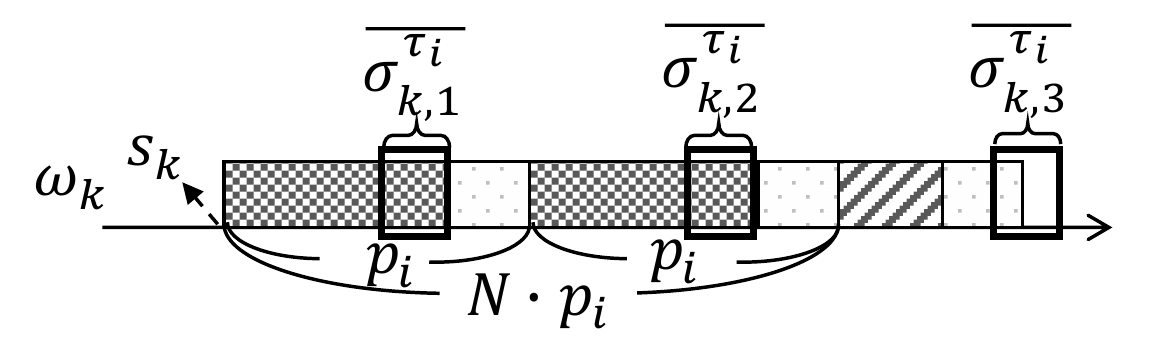}
        \vspace{-1.5\baselineskip}
        \caption{\emph{0-1-interval} has a conflict thus $\NK=N$ is clarified.}
    \end{subfigure}
    \vspace{-0.5\baselineskip}
    \caption{Examples of validating $\NK$ values with applying arrival window $\overline{a_i}$ to busy interval $\omega_k$.}
    \label{fig:arriv_to_nk}
    \vspace{-1\baselineskip}
\end{figure}

\subsubsection{Eliminating Ambiguous Estimates}
\label{approach:remove_mismatched_nk}

As illustrated in Figure~\ref{fig:example_intersect_arrival_window}(b), the ambiguity in arrival windows is caused by \emph{0-1-interval} segments in each arrival window. 
In such a case, if one of the \emph{0-1-interval} segments can be confirmed as \emph{1-interval}, then we know the correct arrival window. Since the estimation of arrival windows relies on \emph{1-interval} and \emph{0-1-interval} segments that are computed based on $\NK$ -- some of which may have ambiguous $N$ or $N+1$ values, the solution is to eliminate the ambiguity in $\NK$. 
This can be done by applying the arrival window $\overline{a_i}$ to each busy interval to validate the estimated $\NK$ values.
%
%
Figure~\ref{fig:arriv_to_nk} illustrates the process of removing mismatched $\NK$ from a busy interval $\omega_k$ with ambiguous $\NK$ values by applying inferred arrival windows $\overline{a_i}$. The $h^{th}$ arrival window $\ATMEW$ of $\tau_i$ in $\omega_k$ is obtained using:

\footnotesize
\vspace{-1.2\baselineskip}
\begin{align}
\ATMEW =  [ a_{i,begin} + \Big(\ceil[\bigg]{ \dfrac{s_k}{p_i} }+h-1 \Big) p_i , 
 a_{i,end} + \Big(\ceil[\bigg]{ \dfrac{s_k}{p_i} }+h-1 \Big)  p_i] \label{equ:job_arrival_window}
\vspace{-0.5\baselineskip}
\end{align}
\normalsize
\noindent
where $a_{i,begin}$ and $a_{i,end}$ are the beginning and the end time points of the initial arrival window $\overline{a_i}$. Figure~\ref{fig:arriv_to_nk} considers the simple case where the arrival window is a continuous interval (shown by the black rectangle). 
That is, the results of the previous step 
involving the computation of a distribution 
are the black rectangles -- as expected, they repeat every $p_i$ seconds. Part (a) of this figure shows the case where the arrival window overlaps with a 0-1-interval. This implies that the 0-1-interval is in fact a 1-interval. 
Therefore, the inference 
becomes $N+1$. Part (b) shows the case that the arrival window does not overlap with a 0-1-interval -- it implies that the 0-1-interval is in fact a 0-interval leaving $N$ as the only possibility.



By removing mismatched $\NK$ values,
the number of possible task combinations for that interval is reduced, sometimes to a unique combination. This reduced set of possible $\NK$ values is then used to update the arrival windows iteratively until the values are stabilized (see Figure~\ref{fig:attack_flow_chart}).

\begin{example}
\label{ex:complete_remove_mismatched_nk}

\noindent
Consider busy interval 
$\omega_4 = [24, 27]$ 
that has two possible task composition inferences for 
$\{\vert J^{\tau_1}_4 \vert, \vert J^{\tau_2}_4 \vert, \vert J^{\tau_3}_4 \vert\}$
in Example~\ref{ex:complete_estimate_nk_1}: $\{1, 0, 1\}$ and $\{1, 1, 0\}$. 
By applying $\overline{a_2}$ to $\omega_4$ as shown in Figure~\ref{fig:example_intersect_arrival_window}(d), we confirm that $\tau_2$ should have arrived one time. This validates the inference of $\{1, 1, 0\}$ and eliminates $\{1, 0, 1\}$.
Once the ambiguity in $\omega_4$ is removed, the only \emph{0-1-interval} for inferring $\overline{a_3}$ in Example~\ref{ex:complete_infer_arrival_windows} can be identified as \emph{0-interval} based on the correct inference $\{1, 1, 0\}$. 
This leaves the interval on the left in Figure~\ref{fig:example_intersect_arrival_window}(b) as the only correct arrival window for $\tau_3$.

\end{example}

\subsubsection{Reconstructing Schedules}
\label{approach:reconstruct_schedules}
\label{approach:arrival_windows_to_arriavl_times}
The final step is to generate the start time $\STME$ of each job to reconstruct the schedule. We do this by feeding $\Gamma$ along with inferred arrival times $a_i$ into a \emph{compact scheduling translator} (to be explained below).
However, at the end of the previously discussed refinement loop, some tasks may still end up with a wider arrival window that cannot be narrowed further. To use a \emph{compact scheduling translator}, an exact point for each arrival $a_i$ is needed. We propose the \emph{beginning point of the first arrival window} as the exact arrival time for such tasks. The reasons are:
\ci to make sure that jobs in a busy interval do not become disconnected and \cii since this choice indicates the earliest possible arrival time of a job, attacks launched using this arrival time will never miss the job.
%
Once each arrival time $a_i$ is determined, 
the \emph{compact scheduling translator} is ready to reconstruct the schedule. For a selected busy interval $\omega_k$, the translator first calculates $\ATME$ for every involved job by using 
$\ATME = a_i + \Big(\ceil[\bigg]{\dfrac{s_k}{p_i}}+h-1 \Big)p_i$
adapted from Equation~(\ref{equ:job_arrival_window}).
%
Then, these $\ATME$ values can be interpreted as a prearranged arrival queue where the scheduling translator only processes the given jobs. The output of this process is the start time $\STME$ of each job within the busy interval $\omega_k$.

Figure~\ref{fig:reconstruct_schedules} presents an example of schedule reconstruction of busy interval $\omega_k$ in the presence of $\tau_i$ and $\tau_{i+1}$ where $\vert J^{\tau_i}_k \vert = 2$, $\vert J^{\tau_{i+1}}_k \vert = 3$ and $pri_i>pri_{i+1}$. Part (a) shows the arrival times of each job that are obtained from arrival windows and (b) presents the start times output from the scheduling translator. By repeating this process for all busy intervals, the full schedule can be reconstructed.

\begin{figure}[h]
    \centering
    \begin{subfigure}[t]{0.49\columnwidth}
        \centering
        \includegraphics[width=0.99\columnwidth]{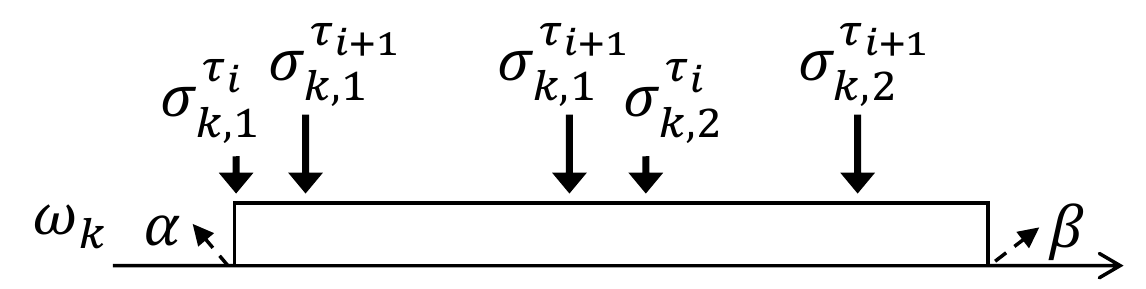}
        \vspace{-1\baselineskip}
        \caption{Inference of arrival times in $\omega_k$.}
    \end{subfigure}%
    \begin{subfigure}[t]{0.49\columnwidth}
        \centering
        \includegraphics[width=0.99\columnwidth]{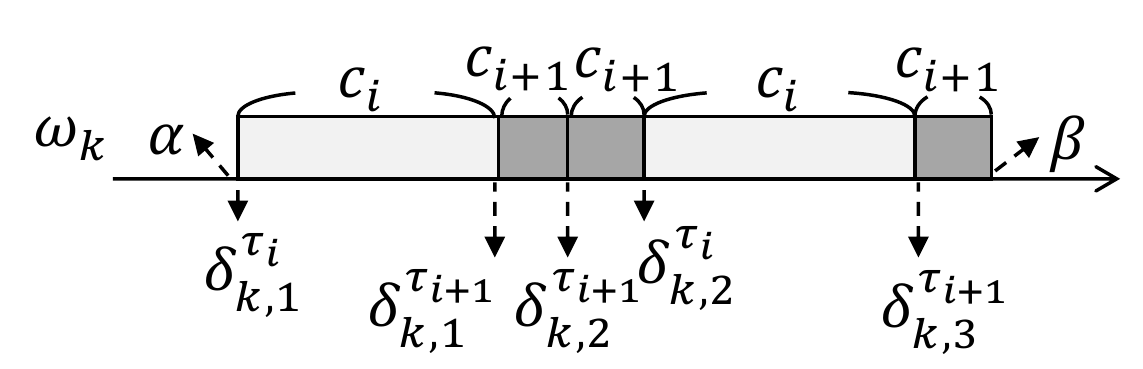}
        \vspace{-1\baselineskip}
        \caption{Translated start times.}
    \end{subfigure}
    \vspace{-0.5\baselineskip}
    \caption{Translation of start times $\STME$ from arrival times $\ATME$ in $\omega_k$ where $pri_i > pri_{i+1}$.}
    \label{fig:reconstruct_schedules}
   \vspace{-1\baselineskip}
\end{figure}


%
%
%
%
%
%
%


\section{Evaluation}
\label{sec::eval}

The \ATTCKNF algorithm has been implemented in a simulation platform as well as on {\em a hardware board running a real-time operating system}. 
The simulator is used to explore a larger design space while the implementation on the hardware board demonstrates the feasibility of carrying out such attacks on realistic systems. 
In this section, we first propose metrics that can be used to evaluate the performance of our algorithms. We then present an evaluation of the algorithm on both simulation engines and hardware board implementation.

\subsection{Performance Metrics}
\label{subsec::metrics}
One way to evaluate the performance of such algorithms is to compare the schedules from the 
output with that from the ground truth. 
Let the actual start times for $\tau_i$ be $\{{\delta^{\tau_i}_1}^*,\cdots,{\delta^{\tau_i}_u}^*\}$ and the estimated start times be $\{\delta^{\tau_i}_1,\cdots,\delta^{\tau_i}_u\}$. 
Let $E(\tau_i)_r={\delta^{\tau_i}_r}^*-\delta^{\tau_i}_r$ be the error in estimating the start time of the $r^{th}$ appearance of task $\tau_i$. Therefore, for task $\tau_i$ we have the errors $E(\tau_i)=\{E(\tau_i)_1,\cdots,E(\tau_i)_u\}$. We define the standard deviation of these errors from zero as
$SD_i=\sqrt{\frac{1}{u}\sum_{r=1}^s(E(\tau_i)_r)^2}$.
%
We can now define the precision of the estimation of the start time of task $\tau_i$ as $(1-\frac{SD_i}{p_i})$ where $p_i$ is the period of task $\tau_i$. This value is a number in $[0,1]$ where $1$ indicates an {\em exact estimation} of the start times.
For a task set containing $n$ tasks, we define the overall 
precision ratio of the algorithm as the 
arithmetic
mean of the estimation precision of tasks in the set, as follows:
\vspace{-0.5\baselineskip}
\small
\begin{equation}
\vspace{-\baselineskip}
\eta'=\frac{1}{n}\sum_{i=1}^n(1-\frac{SD_i}{p_i})
\label{equ:precision_ratio}
\vspace{1\baselineskip}
\end{equation}
\normalsize

\subsection{Simulation-based Evaluation}
\label{subsec:simulation_eval}
\subsubsection{Simulation Setup}
\label{subsec::simulation_setup}
The \ATTCKNF approach was evaluated using 
both 
an internally developed {\em scheduling simulation tool} and real hardware. 
The simulation tool was used to test the scalability of the algorithm and also to test with a more diverse set
of real-time task combinations. 



\noindent
\textbf{Task Set Generation:}
We apply \ATTCKNF to randomly generated synthetic task sets
and check whether the inference of start times matches the corresponding ground truth. The task sets are grouped by utilization from $[0.001+0.1\cdot x,  0.1+0.1\cdot x]$ where $0 \leq x \leq 9$. For example, the $[0.501, 0.6]$ group contains the task sets that occupy $50.1\%$ to $60\%$ of CPU utilization. 
Each utilization group consists of $6$ subgroups that have a fixed number of tasks from $10$ to $15$ respectively. Each subgroup contains $100$ task sets. In other words, $600$ task sets are generated in each utilization group resulting in $6000$ task sets to test for {\em one graph}.
For ease of comparison, we generate task sets with periods that are computed from the factors selected from $[2, 3, 5, 7, 11, 13]$.
Note that The resulting task set may contain some tasks that are harmonic. 
Then, the initial arrival time (\ie task offset) for a task is randomly selected between $0$ and the task's period (\ie $0 \leq a_i < p_i$).

\noindent
\textbf{Scheduling Algorithm:}
We use the commonly used rate-monotonic algorithm \cite{LiuLayland1973} to assign the priorities of tasks,
\ie a task with a shorter period is assigned a higher priority. We only pick those task sets
that are schedulable by fixed-priority scheduling algorithms.



\noindent
\textbf{Execution Time Variation:}
We use normal distribution to produce execution time variation.
A task set is first generated from the aforementioned task set generator that guarantees schedulability.
Then, for a task $\tau_i$, the average execution time $c_i$ is computed by $c_i=wcet_i\cdot 80\%$, where $80\%$ is chosen empirically.
Next, we fit a normal distribution $\mathcal{N}(\mu, \sigma^2)$ for the task $\tau_i$. We let the mean value $\mu$ be $c_i$ and find the standard deviation $\sigma$ with which the cumulative probability $P(X \leq wcet_i)$ is $99.99\%$.
As a result, such a normal distribution produces variation of which $95\%$ are within $\pm 10\% \cdot wcet_i$.
To not violate the schedulability, we adjust the  deviated execution time to be WCET if it exceeds WCET. 

\subsubsection{Simulation Results}
\label{subsec::simulation_results}
We examine the factors affecting the performance of the \ATTCKNF algorithm.
We test the algorithm with task sets generated under the criteria introduced in \cref{subsec::simulation_setup} (unless stated otherwise) and evaluate the precision ratio using Equation~(\ref{equ:precision_ratio}). 

\noindent\textbf{Empirical Baseline of Inference Precision:}
To understand how the precision ratio can reflect the correctness of inferences, we compare the \ATTCKNF algorithm with a naive algorithm -- the initial arrival times are inferred by choosing random time points. 
%
Figure~\ref{fig:naive_algo} shows that the naive algorithm 
yields inference precision ratios ranging from $0.04$ to $0.45$ -- this means that
an inferred start time instant has an error over $55\%$ of the period of the corresponding task when compared to the actual instant. It also implies that an attacker may miss an entire computation interval of the victim task and target the wrong interval.
%
%
%
In contrast, when using the \ATTCKNF algorithm, the inference precision ratio remains high, leading to an average of $0.95$ in precision ratio. This represents an average deviation of $5\%$ (or less) between an inferred start time and the actual start time.


%

\begin{figure}[t]
\vspace{-0.5\baselineskip}
  \centering
  \includegraphics[width=0.75\columnwidth]{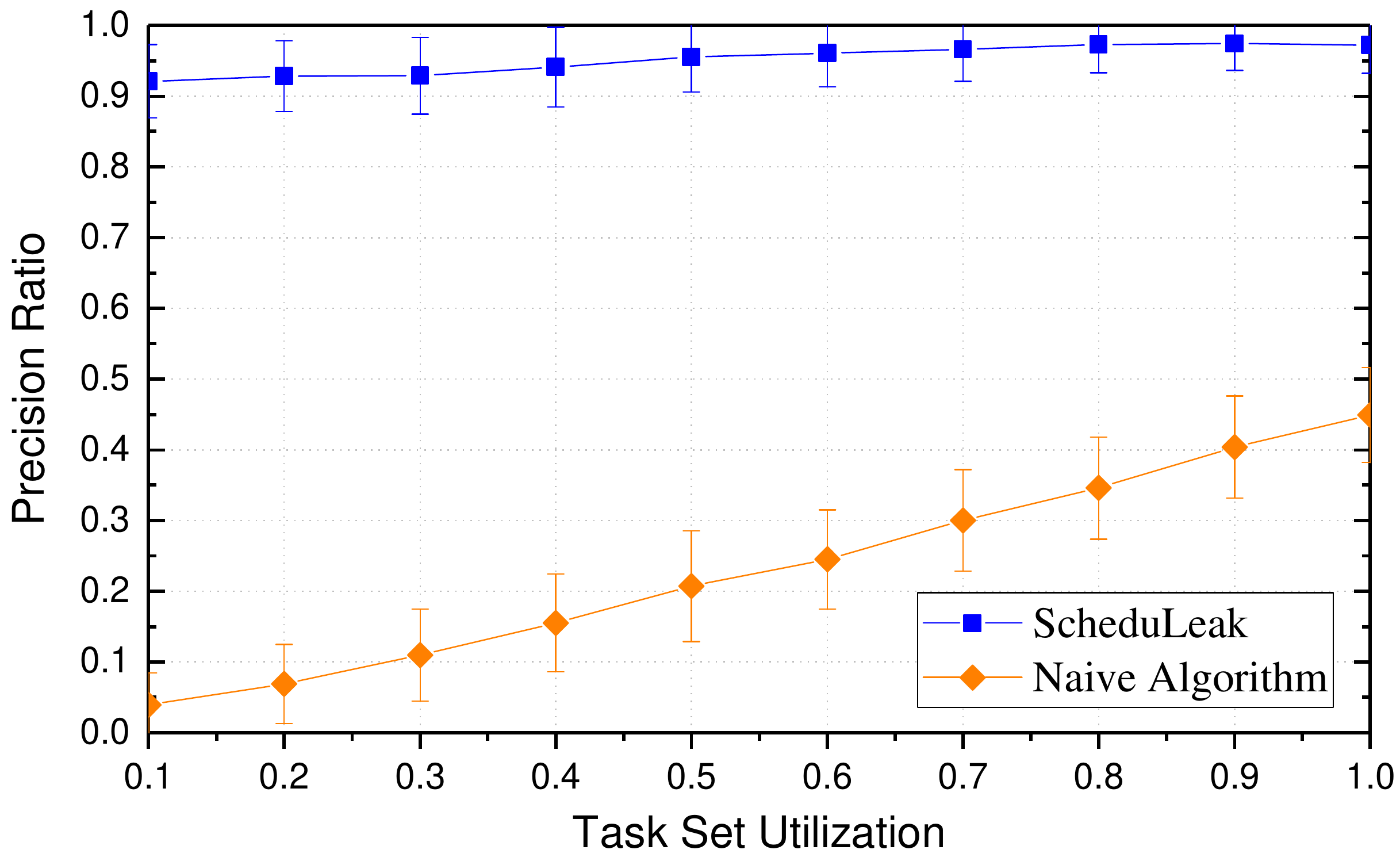}
\vspace{-0.5\baselineskip}
\caption{A naive algorithm that infers initial arrival times by selecting random time points gives an empirical lower bound of the inference precision ratio. X-axis represents the utilization group, and Y-axis is the corresponding precision ratio.}
\label{fig:naive_algo}
\vspace{-1\baselineskip}
\end{figure}

\begin{figure}[t]
  \centering
  \includegraphics[width=0.75\columnwidth]{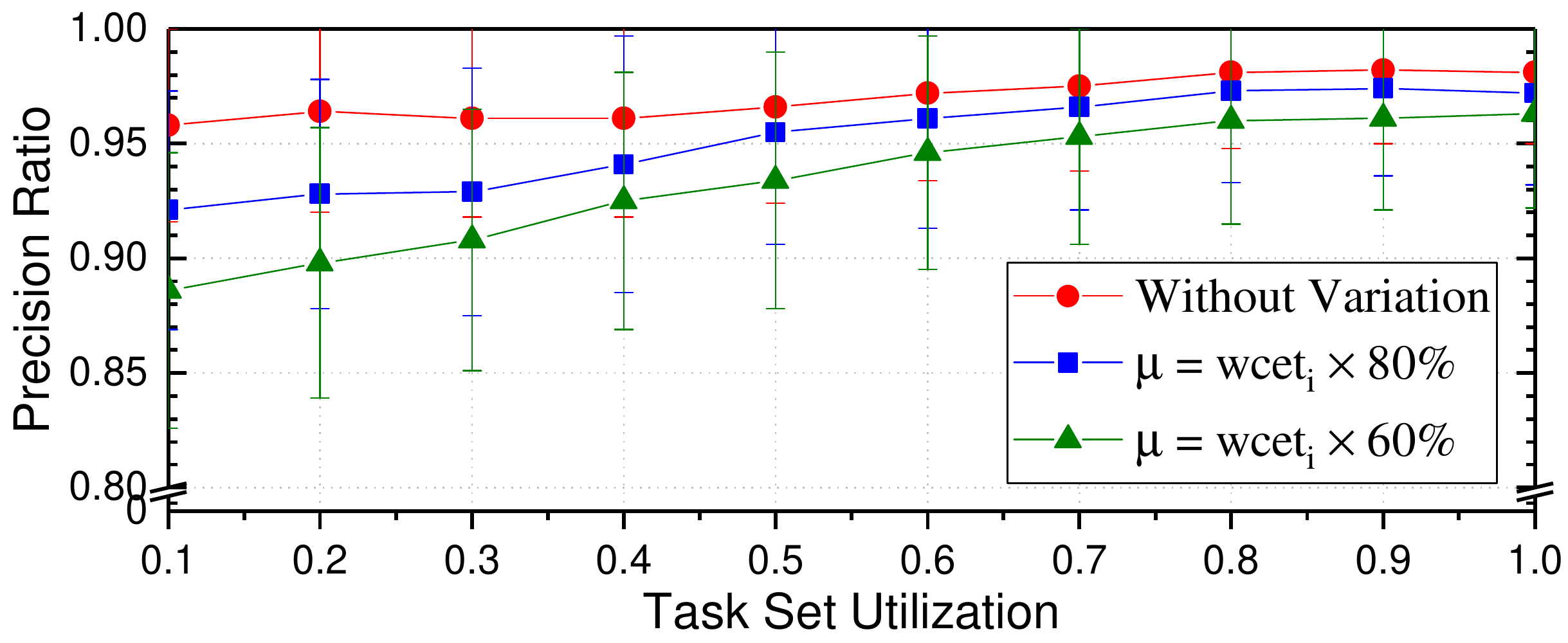}
\vspace{-0.5\baselineskip}
\caption{The comparison of the impact on the inference precision between different degrees of execution time variation. 
The blue line (middle) is the variation configuration used throughout the other experiments in this section. 
}
\label{fig:impact_of_jitters}
\vspace{-0.27in}
\end{figure}


\noindent
\textbf{Impact of Execution Time Variation:}
To understand the extent to which execution time variation impacts the inference precision of the \ATTCKNF algorithm,
we carried out simulations {\em without} execution time variations. We compare such a case with the condition that involves execution time variations produced by the normal distribution where $\mu=wcet_i \cdot 80\%$ and $P(X \leq wcet_i)=99.99\%$ introduced in \cref{subsec::simulation_setup} (this is also the configuration used by other simulation experiments presented in this section).
We also generate another worse case with $\mu=wcet_i \cdot 60\%$ and $P(X \leq wcet_i)=99.99\%$ for comparison.
The results displayed in Figure~\ref{fig:impact_of_jitters},
suggest
that the \ATTCKNF algorithm performs better when there is no execution time variation (the red line). It yields $1.0$ precision ratio in $51.2\%$ of the tested task sets with an overall mean precision ratio of $0.979$.
It's worth nothing that it does not reach $1.0$ precision ratio for all task sets because some may contain harmonic tasks which make it hard to distinguish them when reconstructing.
Since having constant execution times is the best case for \ATTCKNF, we consider this graph an empirical upper bound of the inference precision ratio when a full hyper-period is observed.
Additionally, the figure shows that larger variation leads to lower inference performance.
Nevertheless, it also suggests that the \ATTCKNF algorithm can tolerate wide variations.
It yields an average precision ratio of $0.95$ in $\mu=wcet_i \cdot 80\%$ while it is $0.929$ in $\mu=wcet_i \cdot 60\%$. 

\begin{figure}[t]
\vspace{0.5\baselineskip}
\centering
  \includegraphics[width=0.75\columnwidth]{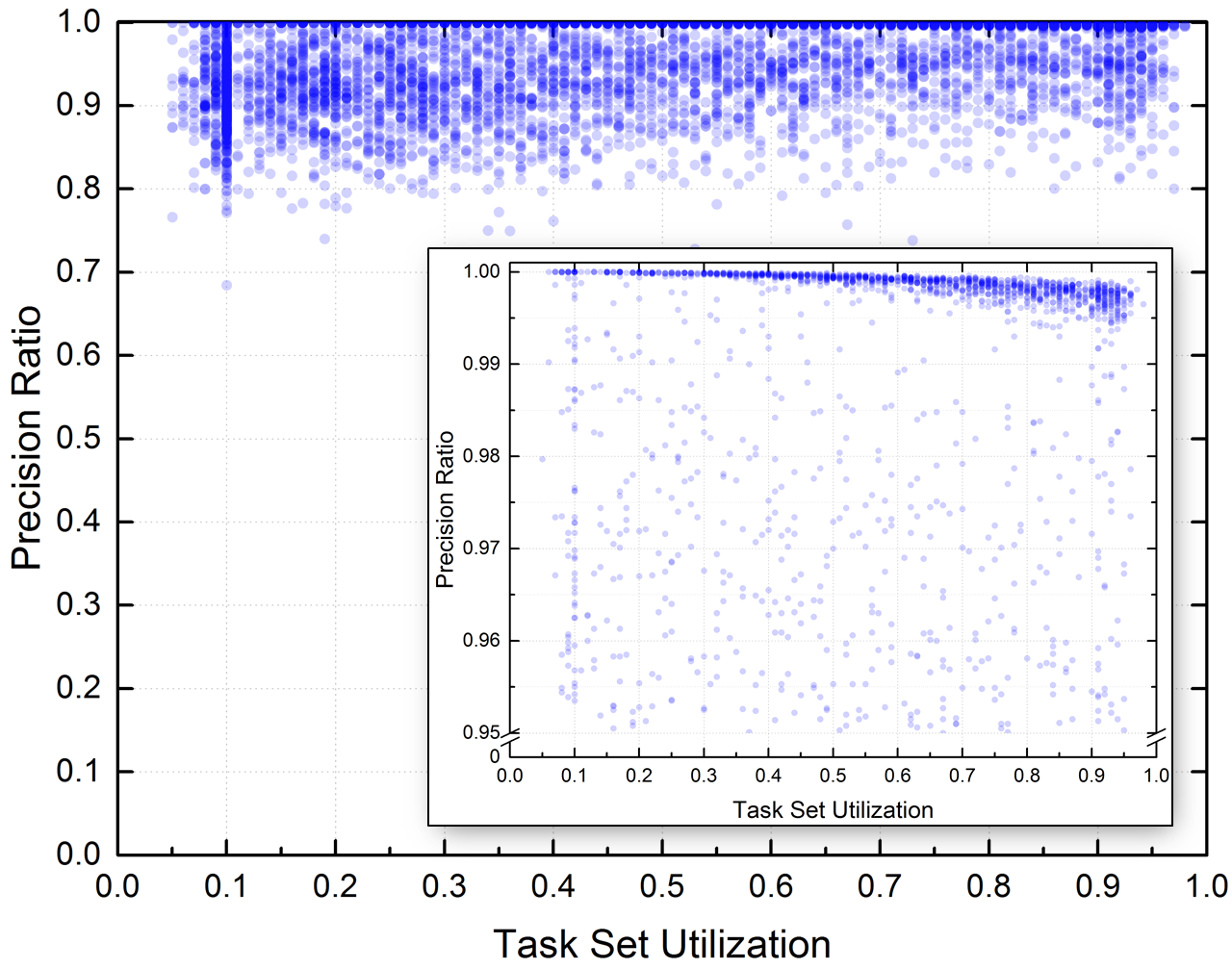}
\vspace{-0.5\baselineskip}
\caption{Detailed data of the 
$\mu=wcet_i \cdot 80\%$ graph in Figure~\ref{fig:impact_of_jitters}. 
X-axis is the utilization of the task sets, and Y-axis is the corresponding inference precision.
A heavier color indicates higher occurrence.}
\label{fig:impact_of_util}
\vspace{-0.9\baselineskip}
\end{figure}

\begin{figure}[t]
\centering
  \includegraphics[width=0.75\columnwidth]{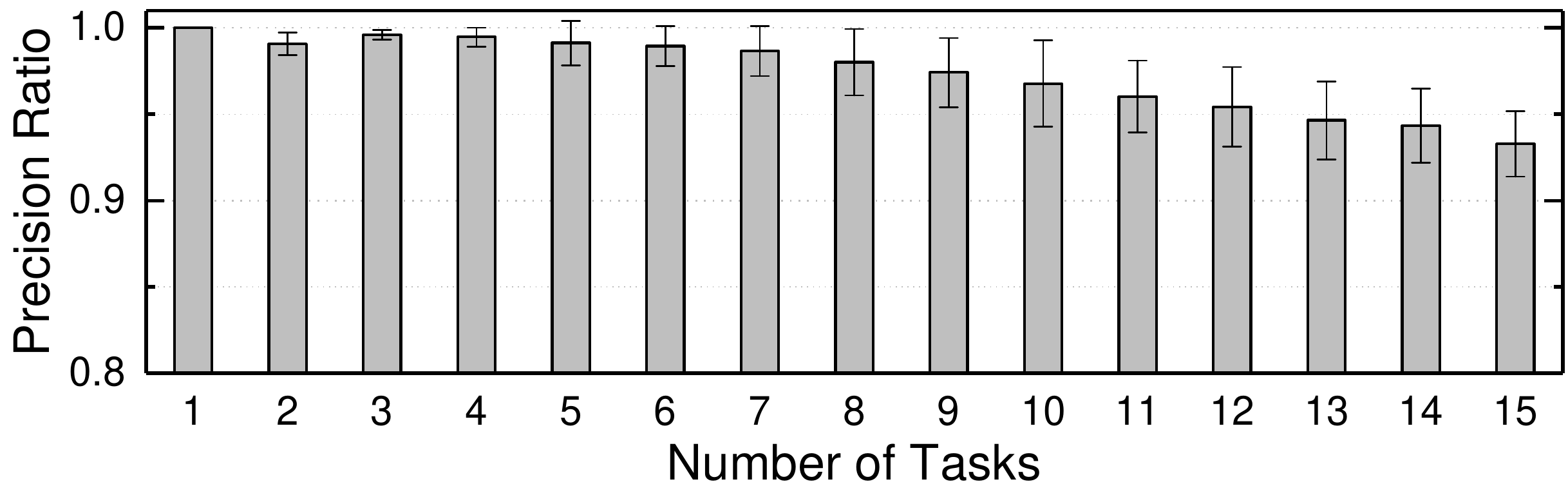}
\vspace{-0.5\baselineskip}
\caption{The inference precision of varied number of tasks. X-axis represents the group of task sets categorized by the number of tasks in a task set. Y-axis is the corresponding inference precision.} 
\label{fig:impact_of_task_num}
\vspace{-0.25in}
\end{figure}

\begin{table}[b]\footnotesize
\vspace{-1\baselineskip}
\centering
\caption{Summary of the inference precision for each utilization group plotted in Figure~\ref{fig:impact_of_util}.}
\label{table:eva_precision_ratio_mean}
\begin{tabular}{|c|c|c|c|c|c|}
\hline
Utilization   & Mean    & SD      & Min    & Median  & Max    \\ \hline
{[}0.0,0.1{]} & 0.9211 & 0.0517 & 0.6846 & 0.9210  & 1      \\ \hline
{[}0.1,0.2{]} & 0.9275 & 0.0498 & 0.7403 & 0.9285  & 1      \\ \hline
{[}0.2,0.3{]} & 0.9290 & 0.0538 & 0.7772 & 0.9289  & 1      \\ \hline
{[}0.3,0.4{]} & 0.9407 & 0.0557 & 0.7499 & 0.9466  & 0.9999 \\ \hline
{[}0.4,0.5{]} & 0.9555 & 0.0495 & 0.8047 & 0.9692  & 0.9999 \\ \hline
{[}0.5,0.6{]} & 0.9606 & 0.0483 & 0.7281 & 0.9832  & 0.9998 \\ \hline
{[}0.6,0.7{]} & 0.9665 & 0.0447 & 0.7576 & 0.9977  & 0.9997 \\ \hline
{[}0.7,0.8{]} & 0.9728 & 0.0405 & 0.7386 & 0.9975  & 0.9996 \\ \hline
{[}0.8,0.9{]} & 0.9737 & 0.0379 & 0.8008 & 0.9968  & 0.9993 \\ \hline
{[}0.9,1.0{]} & 0.9722 & 0.0398 & 0.8005 & 0.9958  & 0.9991 \\ \hline
\end{tabular}
\end{table}

\noindent
\textbf{Impact of Task Set Utilization:}
Next, we analyze how the {\em utilization} of the real-time task sets can affect the precision of our analysis.
The precision ratios for each utilization group are plotted in Figure~\ref{fig:impact_of_jitters} (the blue line).
Raw precision ratios for each task set are provided in Figure~\ref{fig:impact_of_util}, and the statistical data is summarized in Table~\ref{table:eva_precision_ratio_mean}.
From the figures, we observe that the precision ratio decreases as the utilization decreases. 
It is because lower utilizations may indicate \ci potentially shorter execution times and \cii smaller and scattered busy intervals -- both of which are detrimental to our analyses. 
If a task's execution time is smaller than any other task's variation, it can cause the algorithm to falsely filter out such a short task when decomposing a busy interval, resulting in inaccurate inferences. Further, small and scattered busy intervals reduce the probability that an inferred arrival time can successfully land within the busy interval. 
Under such circumstances, \ATTCKNF can't resolve certain inferences.
This false negative is due to the fact that the algorithm only considers the beginning point of an inferred arrival window when reconstructing the schedule of a busy interval (as stated in \cref{approach:reconstruct_schedules}).
This ambiguity influences the inference precision the most when the busy interval is short.

Furthermore, from the zoomed diagram in Figure~\ref{fig:impact_of_util} and Table~\ref{table:eva_precision_ratio_mean} we also observe that the maximum and median precision ratios decrease slightly at high utilizations. It is because higher utilizations lead to larger busy intervals. Consequently more ambiguous conditions show up while trying to estimate the number of arrivals in a busy interval, making it difficult for the \ATTCKNF algorithm to arrive at a correct inference.




\noindent
\textbf{Impact of the Number of Tasks:}
Another factor that may affect the precision ratio is {\em the number of tasks} in a task set.
We compute mean precision ratios for 15 subgroups that have the number of tasks ranging from 1 to 15 and see whether it influences the precision. The result are presented in Figure~\ref{fig:impact_of_task_num}. 
It shows that the more tasks involved in a system, the harder the \ATTCKNF algorithm can correctly infer the schedules. 
It is because having more tasks in a system means potentially more tasks may be involved in a busy interval. It also implies that there can be more possible combinations for the composition of a busy interval, which increases the difficulty of the analysis from \ATTCKNF.

\noindent
\textbf{Impact of Observation Duration:}
In the previous experiments, we fed \ATTCKNF with busy intervals observed from a full hyper-period.
However, since busy intervals are analyzed 
in a discrete fashion
in \ATTCKNF, it is also possible to process the algorithms with busy intervals observed from fewer or more than a full hyper-period.
To evaluate the impact of {\em observation duration}, we let \ATTCKNF observe different proportions of the schedule and collect busy intervals within the observation duration.
The busy intervals that are not complete at the edge are discarded.
The proportion we examine here ranges from $10\%$ to $200\%$ of a hyper-period.
The results are shown in Figure~\ref{fig:hp_vs_rp}.
From the figure, we can see that \ATTCKNF yields a precision ratio of $0.58$ when observing $10\%$ of a hyper-period. 
If we only compute the corresponding precision ratio for the busy intervals within the observation duration, we get a precision ratio of $0.77$ when observing $10\%$ of a hyper-period.
In both cases, the inference precision gets higher as more portions of the schedule are observed.
It's worth mentioning that the inference precision is improved even after collecting over $100\%$ of a hyper-period.
It is due to the nature of normal distributions where the variation around the mean point has high occurrence probability.
That is, as more data is sampled, the variation can be neutralized, resulting in the \ATTCKNF algorithm's ability to infer more accurate arrival windows and start points.

\begin{figure}[t]
  \centering
  \includegraphics[width=0.75\columnwidth]{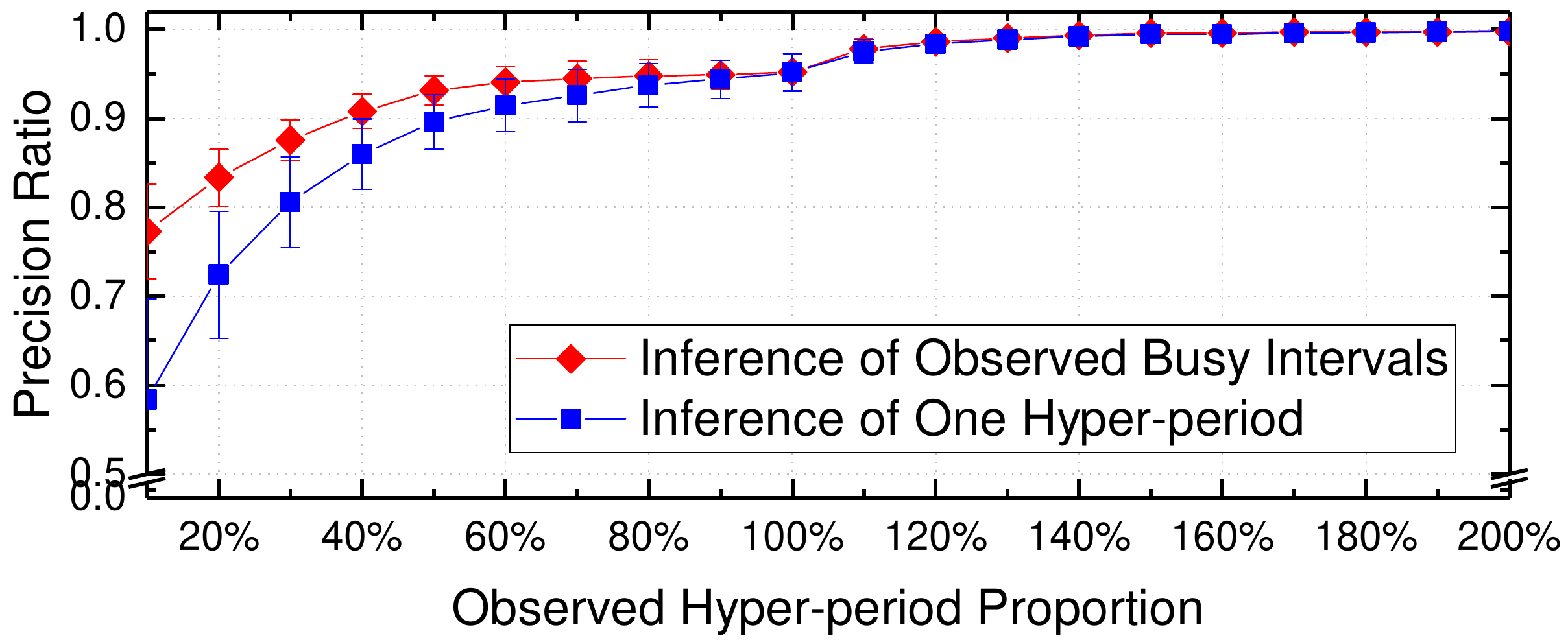}
\vspace{-0.5\baselineskip}
\caption{Impact of observation duration on precision ratio.
The precision ratio gets improved as more schedule data is collected.}
\label{fig:hp_vs_rp}
\vspace{-1\baselineskip}
\end{figure}

\subsection{Zedboard-based Evaluation}
\label{subsec:zedboard_eval}

\subsubsection{Zedboard Setup}
We implemented the \ATTCKNF algorithm on an ARM-based
development board, the {\bf Zedboard}\cite{Zedboard}. 
The processor runs at $666.67MHz$ and drives a shared $64$-bit \emph{Global Timer (GT)} that is clocked at $\approx 333.3MHz$.
We tested our implementation with a variety of different task sets, starting from the UAV model (Figure \ref{fig:sys_example} and \cref{subsec:example_case_study}) to many synthetic task sets. 
The task sets ran on a real-time operating system, \emph{FreeRTOS} -- a lightweight and open source real-time kernel\cite{FreeRTOS}. It is powered by a priority-based preemptive real-time scheduler. 
To obtain ``ground truth'', 
the FreeRTOS kernel has been modified to record time stamps for context switches. 



\subsubsection{Implementation}
\label{implementation}

We introduce how the \OBSERVER captures busy intervals as this step is hardware-dependent. The \emph{Global Timer (GT)} that is accessible to user tasks in \emph{FreeRTOS} is used as a time reference to capture busy intervals. As mentioned above, \emph{GT} is clocked at $333.3MHz$, hence its resolution can be computed as $1/333.3MHz \approx 3ns$. To capture a busy interval, the \OBSERVER uses \emph{GT} as a reference to inspect whether or not it has been preempted. The \OBSERVER reads \emph{GT} in a loop and compares it with the time stored from last read. If nothing preempts the \OBSERVER during the loop, then the difference for times between consecutive loops remains below one unit of the loop execution time (time measurement is presented in \cref{subsec::zedboard_results}). In contrast, if the difference is greater than expected, then a preemption must have occurred. 
%
The \OBSERVER repeats the same process until it collects all busy intervals for at least one hyper-period. After that, the program moves to the busy interval analysis stage. The granularity of the \OBSERVER measurements is further discussed in the next part.

\subsubsection{Zedboard-based Results}
\label{subsec::zedboard_results}

\textbf{Observer Task Overhead:}
We first examine the execution cost of the \OBSERVER and its impact on the ability to capture
busy intervals. Our \OBSERVER uses a loop to read time values from the \emph{Global Timer}. The shortest busy interval that can be measured depends on the cost of reading the timer and executing the \OBSERVER instructions. 
The cost of each read loop measures $447 ns$, on average, in the absence of preemption. 
This is really small compared to the execution times of real-time tasks which are typically more than $10 \mu s$. 
The \emph{Global Timer} counts $447 ns/3 ns \approx 149$ times during each read loop. On the other hand, a small busy interval with just one task that has execution time of $10 \mu s$ will result in the counter incrementing $10 \mu s/3 ns \approx 3333$ times. Hence, it is very unlikely that the \OBSERVER will miss any busy intervals.


The actual cost for capturing a busy interval varies since the timing now involves a pair of {\em context switches}. 
From our experiments a pair of context switches takes $18.03 \mu s$ on average on this board.
This adds to the length of any busy interval but the costs are
bounded as each measurement only includes two read loops. We can confidently remove these costs from every
busy interval. This leaves only {\em jitters} as a source of uncertainty in our measurements. From
our experiments, removing the \OBSERVER costs results in a $0.28 \mu s$ error, on average,
for each busy interval. If the delta due to this error is greater than the execution time for 
the shortest task then it introduces uncertainty into our analysis since we cannot separate it
from a legitimate task. In practice, this depends on the actual application.
For the UAV model presented earlier on, the shortest task (the Network Manager) has an 
execution time of $30 \mu s$ -- two orders of magnitude higher than the error/delta. Hence, the
errors in measuring the busy intervals did not really impact our analyses.

\noindent
\textbf{On-board Analysis Overhead:}
It is important to note that the analysis of busy intervals need not necessarily be performed online. For some attack scenarios, the data can be {\em analyzed offline}. Hence, the analysis will not be limited by the performance of the hardware. In this section
though, we focus on the overheads for carrying out the analysis on the actual board to demonstrate
the feasibility of online analysis.

The estimate of $\NK$ uses Equation~(\ref{eqn:biLength}) to find the matching combinations for a busy interval.
This has order of $2^n$ time complexity. However, the number of tasks is often fixed for a task set, thus $2^n$ is bounded in a given real-time system.
On Zedboard, a $10$-task task set costs $58.94ms$ in the worst case to process all possible combinations for a busy interval.
Calculating arrival windows
depends on the number of job instances of each task in a hyper-period. 
From our experiments on average, it takes $2.073ms$ to calculate arrival windows for the $10$-task task set mentioned above where each task executes for around $10.5 \mu s$.
%
For the elimination of mismatched $\NK$ values it takes arrival windows of $n$ tasks to inspect each estimate in every busy interval. 
This has a complexity in the order of $n$. 
In the same experiment as above, it takes $17.4us$ to iterate through $14$ busy intervals with $10$ arrival windows from the correspondent $10$ tasks.
Finally, for a task set with $10$ tasks that have $14$ busy intervals on our \emph{Zedboard}, it takes $828.05ms$ to complete the total analysis.
%
In this test case, \ATTCKNF yields a precision ratio of $0.9977$.

\subsubsection{Implications of \ATTCKNF}
\label{sub:eval_cache_attack}
We now demonstrate an attack that takes advantage of the task schedule information obtained using \ATTCKNF by implementing the attack case introduced in \cref{subsec:example_case_study}. 
In particular we use a cache-based side-channel attack \cite{kelsey1998side, osvik2006cache, page2002theoretical} 
(as an instance of storage-channel-based timing attacks) to precisely gauge 
the cache usage of a victim task, in this case the image encoder task.
The cache-based side-channel attack is launched from a periodic task that has the highest priority (as opposed to the observer task that uses the lowest priority).
With the schedule information, we carefully place the cache {\em prime and probe} steps \cite{osvik2006cache} before and after the image encoder task begins and ends. 
We let the attacker task halt the attack if we know that the victim task will not appear in the next period.
Figure~\ref{fig:cache_attack}(a) displays the result of this 
attack.
It shows that the cache-based side-channel attack with schedule information can effectively filter out unnecessary 
information and preserve the cache usage behavior of the victim task.
On the other hand, Figure~\ref{fig:cache_attack}(b) demonstrates the case where the attacker has no knowledge of the task schedule. The attacker then has to  launch the cache-based side-channel attack in every period randomly. As a result, the obtained cache usage contains too much noise from other tasks, making it difficult to obtain useful information.

\begin{figure}[t]
    \centering
        \includegraphics[width=0.75\columnwidth]{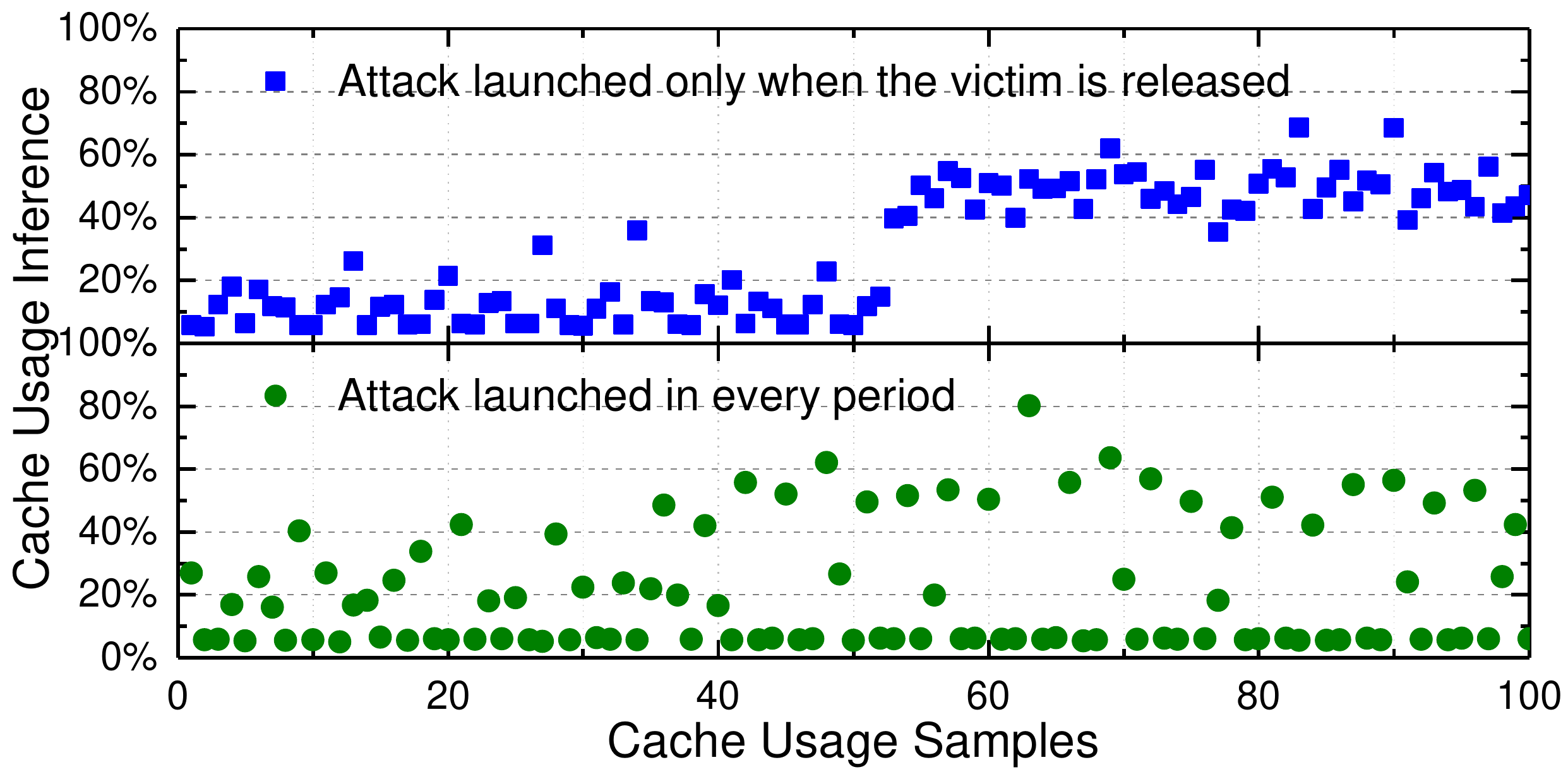}
        \vspace{-0.5\baselineskip}
    \caption{A demonstration of the advantage of a targeted attack. X-axis represents sample points and Y-axis is the estimated cache usage.
Top part shows a successful cache side-channel attack targeting a specific victim task with the knowledge of precise task schedule. 
Bottom part shows that a careless attack leads to unusable information due to too much noise.}
    \label{fig:cache_attack}
      \vspace{-0.2in}
\end{figure}

\section{Discussion}
\label{discussion}

We evaluated the algorithm by varying four variables: \ci execution time variation, \cii CPU utilization, \ciii the number of tasks \civ observed duration. 
From the results, we can conclude that the following factors may influence the inference precision:
\noindent
\textbf{Execution Time Variation:}
Large execution time variation can reduce the precision of the inference. 
As shown in Figure~\ref{fig:impact_of_jitters}, the \ATTCKNF algorithms yield an average precision ratio of $0.95$ and $0.929$ in $\mu=wcet_i \cdot 80\%$ and $\mu=wcet_i \cdot 60\%$ respectively.
\noindent
\textbf{The Number of Tasks:}
Having more tasks in a system means more possible combinations in decomposing a busy interval.
The precision ratio decreases as the number of tasks increases, as shown in Figure~\ref{fig:impact_of_task_num}.
\noindent
\textbf{Busy Interval Length:}
A short busy interval reduces the probability that an inferred arrival time can successfully land within the busy interval.
This mostly happens when the task set utilization is low, which decreases the inference precision.
\noindent
\textbf{Task Set Utilization:}
Higher utilization represents larger busy intervals and more tasks involved in one individual busy interval, leading to increased uncertainty.
This results in a decrease of overall inference precision.
%
\noindent
\textbf{Observed Duration:}
The \ATTCKNF algorithms can work 
with capturing just a portion of the task schedule.
Accumulating more busy interval observation can improve the inference precision.
This is also true when more than one hyper-period is observed.

\section{Related Work}
\label{sec::related}

Recent work on RTS security has demonstrated the feasibility of information leakage via  scheduling. Son \emph{et al.} \cite{embeddedsecurity:son2006} highlighted the susceptibility of RM schedulers to covert timing channel attacks. V\"{o}lp \emph{et al.} \cite{embeddedsecurity:volp2008} presented modifications to fixed-priority scheduling, altering thread blocks that potentially leak information with the idle thread to avoid the exploitation of timing channels. V\"{o}lp \emph{et al.} \cite{embeddedsecurity:volp2013} examined shared-resource covert channels in real-time schedules and addressed it by modifying resource locking protocols.

Kadloor \emph{et al.} \cite{kadloor2013} introduces a methodology for quantifying side-channel leakage for first-come-first-serve and time-division-multiple-access schedulers. Gong and Kiyavash \cite{GongK14} analyzed deterministic work-conserving schedulers.
They discovered a lower bound for the total information leakage. 
The collaborative version of this problem where two users form a covert 
channel in a shared scheduler to steal private information from a secure system is studied by Ghassami \emph{et al.} \cite{ghassami2015capacity}.
While in these works the attacker uses traffic analysis to obtain information about the user activities, our work is primarily concerned with reconstructing the original task schedule as a stepping to stone to other stealthy attacks.



Mohan \emph{et al.} \cite{embeddedsecurity:mohan2014} considered the problem of direct information leakage between real-time tasks through architectural resources such as shared cache. They introduced a modified fixed-priority scheduling algorithm that integrates security levels into scheduling decisions. Pellizzoni \emph{et al.} \cite{embeddedsecurity:mohan2015} extended the above scheme to a more general task model and also proposed an optimal priority assignment method that determines the task preemptibility.


Cache-based side-channels contain the highest bandwidth among all side-channels, making them invaluable for information leakage \cite{
kelsey1998side, osvik2006cache, page2002theoretical}. The growth of cloud computing has caused such attacks face increased scrutiny (\eg \cite{Apecechea_finegrain, ristenpart2009hey}).
Wang and Lee \cite{Wang:2007} presented a hardened cache model to mitigate side-channels by adopting partitioning and memory-to-cache mapping randomization techniques. However, most existing methods aimed at reducing leakage through cache-based side-channels have not considered real-time systems. While other methods of side-channel attacks exist such as power, electromagnetic and frequency analysis \cite{Tiu05anew,Agrawal:2002:ES,Jiang2014}, they are not the focus of this paper. 

\section{Conclusion}
\label{sec::concl}

The methods presented here will improve the design of future
real-time systems. Designers will have an increased awareness of attack
mechanisms that leak crucial information in such systems. Hence, they
can develop methods that prevent or at least reduce the effectiveness of
such attacks. 

We intend to further refine the algorithms 
presented here
and also develop solutions to deter such attacks. We believe that the
area of security for real-time systems will require the development of a
large body of work, along multiple directions, to ensure the overall safety 
of such systems.




\bibliographystyle{abbrvcy}

\end{document}